\documentclass[11pt]{article}
\usepackage[margin=1in]{geometry}

\usepackage{amsmath}
\usepackage{amssymb}
\usepackage{amsthm}
\usepackage{bbm}
\usepackage{mathtools}
\usepackage[shortlabels]{enumitem}
\usepackage{natbib}
\usepackage[colorlinks=true,linkcolor=blue,citecolor=blue]{hyperref}
\usepackage{adjustbox}
\usepackage{graphicx}
\usepackage[title]{appendix}
\usepackage{xcolor}
\usepackage{authblk}
\usepackage{booktabs}
\usepackage{makecell}
\usepackage{multirow}
\usepackage[justification=centering]{caption}

\newtheorem{theorem}{Theorem}

\newtheorem{corollary}{Corollary}

\theoremstyle{remark}
\newtheorem{remark}{Remark}
\theoremstyle{definition}
\newtheorem{definition}{Definition}
\newtheorem{assumption}{Assumption}
\newtheorem{example}{Example}
\newtheorem{setting}{Setting}[example]

\def\E{\mathbb{E}}
\def\R{\mathbb{R}}
\def\D{\mathcal{D}}
\def\F{\mathcal{F}}
\def\H{\mathcal{H}}
\def\I{\mathcal{I}}
\def\M{\mathcal{M}}
\def\N{\mathcal{N}}
\def\U{\mathcal{U}}
\def\Y{\mathcal{Y}}

\def\iid{i.i.d.}
\def\ind{\perp\!\!\!\perp}
\def\convd{\overset{d}{\longrightarrow}}
\def\convp{\overset{p}{\longrightarrow}}
\DeclareMathOperator{\var}{var}
\DeclareMathOperator{\cov}{cov}

\def\rhotilde{\widetilde{\rho}}

\providecommand{\keywords}[1]{{\small\textit{Keywords:} #1}}

\begin{document}

\title{Distance and Kernel-Based Measures for Global and Local Two-Sample Conditional Distribution Testing}
\author[1]{Jian Yan}
\author[2]{Zhuoxi Li}
\author[3]{Xianyang Zhang}
\date{}
\affil[1]{Department of Statistics and Data Science, Cornell University}
\affil[2]{Department of ISOM, Hong Kong University of Science and Technology}
\affil[3]{Department of Statistics, Texas A\&M University}

\maketitle

\begin{abstract}
Testing the equality of two conditional distributions is crucial in various modern applications, including transfer learning and causal inference. Despite its importance, this fundamental problem has received surprisingly little attention in the literature, with existing works focusing exclusively on global two-sample conditional distribution testing. Based on distance and kernel methods, this paper presents the first unified framework for both global and local two-sample conditional distribution testing. To this end, we introduce distance and kernel-based measures that characterize the homogeneity of two conditional distributions. Drawing from the concept of conditional U-statistics, we propose consistent estimators for these measures. Theoretically, we derive the convergence rates and the asymptotic distributions of the estimators under both the null and alternative hypotheses. Utilizing these measures, along with a local bootstrap approach, we develop global and local tests that can detect discrepancies between two conditional distributions at global and local levels, respectively. Our tests demonstrate reliable performance through simulations and real data analysis. 
\end{abstract}

\section{Introduction}
The canonical setting for (nonparametric) two-sample testing focuses on assessing the equality of two unconditional distributions. In many contemporary applications, however, people are instead interested in testing the equality of two conditional distributions. Consider two independent data sets $\{(Y_{i}^{(l)},X_{i}^{(l)})\}_{i=1}^{n_{l}}$ for $l=1,2$. Here $Y_{i}^{(l)}\in\Y$ and $X_{i}^{(l)}\in\R^{p}$, where $\Y$ is allowed to be a general metric space (see Remark \ref{remark:y}). For $l=1,2$, assume that $\{(Y_{i}^{(l)},X_{i}^{(l)})\}_{i=1}^{n_{l}}$ are independent and identically distributed (\iid) samples of $(Y^{(l)},X^{(l)})\sim P^{(l)}\equiv P_{Y\mid X}^{(l)}\otimes P_{X}^{(l)}$, where $P_{X}^{(l)}$ is the marginal distribution of $X^{(l)}$ and $P_{Y\mid X=x}^{(l)}$ is the conditional distribution of $Y^{(l)}$ given $X^{(l)}=x$. To ensure that the testing problem considered below is nontrivial, we assume that $P_{X}^{(1)}$ and $P_{X}^{(2)}$ are equivalent, i.e., $P_{X}^{(1)}\ll P_{X}^{(2)}$ and $P_{X}^{(2)}\ll P_{X}^{(1)}$, where $P \ll Q$ means that $P$ is absolutely continuous in reference to $Q$. We aim to test the following hypothesis, which we call the global two-sample conditional distribution testing problem,
\begin{equation}\label{eq:hypo}
    H_{0}:P^{(1)}_{X}(P_{Y\mid X}^{(1)}=P_{Y\mid X}^{(2)})=1\quad\text{versus}\quad H_{a}:P^{(1)}_{X}(P_{Y\mid X}^{(1)}\neq P_{Y\mid X}^{(2)})>0. 
\end{equation}
Due to the equivalence between $P_{X}^{(1)}$ and $P_{X}^{(2)}$, the hypotheses in (\ref{eq:hypo}) can be equivalently formulated by replacing $P^{(1)}_{X}$ with $P^{(2)}_{X}$. 

We want to emphasize that $Y$ and $X$ denote two generic random variables and do not necessarily correspond to response and covariates, respectively. For instance, in the prior shift example below, $Y$ are covariates, and $X$ is a response in (\ref{eq:hypo}). Moreover, the marginal distributions of $X$ from the two populations may differ (i.e., $P_{X}^{(1)}\neq P_{X}^{(2)}$), as seen in the motivating examples below. Thus, $H_{0}$ in (\ref{eq:hypo}) is not equivalent to $P^{(1)}=P^{(2)}$, and unconditional two-sample tests for the equality of two joint distributions are generally not applicable in our context. Applying such tests would result in a failure to control the type I error when $P_{X}^{(1)}\neq P_{X}^{(2)}$. Furthermore, we mainly focus on scenarios in which both $Y$ and $X$ are continuous. When $Y$ is categorical, (\ref{eq:hypo}) reduces to testing the equality of conditional means. Conversely, when $X$ is categorical, (\ref{eq:hypo}) becomes equivalent to unconditional two-sample testing problems. 

Hypothesis (\ref{eq:hypo}) is central to many important problems in econometrics, machine learning, and statistics. For example, in transfer learning, the prior and covariate shift assumptions are commonly employed to tackle distributional differences between source and target populations \citep{kouw2018introduction}. The prior shift assumption asserts that the conditional distribution of the covariates given the response is identical in both populations while allowing for a shift in the marginal distributions of the response. Conversely, the covariate shift assumption posits that the conditional distribution of the response given the covariates remains invariant across source and target populations, but the marginal distributions of the covariates can differ. Both assumptions are widely adopted in the literature; see, e.g., \citet{huang2024efficient,lee2024doubly} for the prior shift assumption, and \citet{shimodaira2000improving,tibshirani2019conformal,liu2023augmented,ma2023optimally} for the covariate shift assumption. Despite their prevalence, there is a paucity of work that formally validates these assumptions. Both of them can be framed as testing the equality of two conditional distributions as in (\ref{eq:hypo}). Such tests are essential to the validity of methods developed under the prior or covariate shift assumptions. 

Another motivating example comes from causal inference. Testing hypotheses in the context of treatment effect analysis has always been of interest \citep[Sections 3.3 and 5.12]{imbens2009recent}. Consider the standard setup based on the potential outcome framework \citep{rubin1974estimating}. Suppose $\{(Y_{i},T_{i},X_{i})\}_{i=1}^{n}$ are \iid\ observations of $(Y,T,X)$, where $Y$ is the observed outcome of interest, $T$ denotes a binary treatment (1: treated, 0: untreated), and $X$ are pretreatment covariates. For each subject, we define a pair of potential outcomes, $\{Y(1),Y(0)\}$, that would be observed if the subject had been given treatment, $Y(1)$, and control, $Y(0)$. One may be interested in testing the null hypothesis that the conditional distribution of $Y(1)\mid X$ is the same as that of $Y(0)\mid X$ \citep{imbens2009recent}, i.e., zero conditional distributional treatment effect. Under the prevalent assumptions of consistency, $Y=TY(1)+(1-T)Y(0)$, and no unmeasured confounding, $\{Y(1),Y(0)\}\ind T\mid X$, the conditional distributions of $Y(1)\mid X$ and $Y(0)\mid X$ can be identified as those of $Y\mid X,T=1$ and $Y\mid X,T=0$, respectively. Therefore, it can be formulated as testing hypothesis (\ref{eq:hypo}) with $\{(Y_{i},X_{i}):T_{i}=1\}$ and $\{(Y_{i},X_{i}):T_{i}=0\}$ being the two sets of independent samples. 

Global two-sample conditional distribution testing concerns the entire support of the conditioning variable. However, in many empirical applications, the primary interest lies in a subset of the population defined by a specific value of the conditioning variable \citep{bugni2025testing}. For example, in regression discontinuity designs, researchers seek to compare outcome distributions conditional on the cutoff of the running variable \citep{shen2016distributional}. Additional examples can be found in \citet{bugni2025testing}. Formally, for a fixed value $x$ in the support of $P_{X}^{(1)}$ (or equivalently $P_{X}^{(2)}$), we are interested in testing
\begin{equation}\label{eq:hypo2}
    H_{0}:P_{Y\mid X=x}^{(1)}=P_{Y\mid X=x}^{(2)}\quad\text{versus}\quad H_{a}:P_{Y\mid X=x}^{(1)}\neq P_{Y\mid X=x}^{(2)},
\end{equation}
which we refer to as the local two-sample conditional distribution testing problem, in contrast to the global problem (\ref{eq:hypo}). This problem is novel and is partially motivated by \citet{duong2013local} and \citet{kim2019global} in the context of unconditional two-sample testing. Notably, with a little extra effort, our proposed framework can readily handle the local testing problem (\ref{eq:hypo2}). 

Recently, distance and kernel-based measures have received considerable attention in both the statistics and machine learning communities \citep{szekely2017energy,muandet2017kernel}. These measures have been applied to a wide range of hypothesis testing problems, including two-sample testing \citep{szekely2004testing,baringhaus2004new,gretton2012kernel}, goodness-of-fit testing \citep{szekely2005new,balasubramanian2021optimality}, independence testing \citep{szekely2007measuring,gretton2007kernel,chakraborty2019distance,ke2020expected,deb2020measuring} and conditional independence testing \citep{fukumizu2007kernel,wang2015conditional,sheng2023distance}. 
This work aims to establish a distance and kernel-based framework to tackle both the global and local two-sample conditional distribution testing problems. To achieve this, we introduce the conditional energy distance (\ref{eq:ced}) and the conditional maximum mean discrepancy (\ref{eq:cmmd}), along with their integrated version (\ref{eq:i}), which fully characterize the homogeneity of two conditional distributions. Additionally, we show the equivalence between the conditional energy distance and the conditional maximum mean discrepancy. Building on estimators of these measures, we develop global and local tests capable of detecting discrepancies between two conditional distributions at global and local levels, respectively. 

Our estimation strategy employs a combination of U-statistics and kernel smoothing, initially introduced in the so-called conditional U-statistics by \citet{stute1991conditional} and later applied to different problems by \citet{wang2015conditional} and \citet{ke2020expected}. To highlight our theoretical contributions, we summarize the distinct asymptotic behaviors of our global and local test statistics under both the null and alternative hypotheses in Table \ref{tab:asy}. It is noteworthy that local tests based on distance and kernel measures have not been previously explored in the literature. Furthermore, while \citet{wang2015conditional} and \citet{ke2020expected} only provided the asymptotic distributions of their global test statistics under the null, we offer additional insights by systematically studying the properties of our statistics under both the null and alternative hypotheses. As a side note, we identify certain gaps in the derivations of the asymptotic null distribution in both \citet{wang2015conditional} and \citet{ke2020expected}, with details provided in Section \ref{sec:global}. 

Our theoretical analysis reveals two key effects of applying kernel smoothing to U-statistic estimators of distance and kernel-based measures. First, kernel smoothing introduces bias into these estimators, necessitating the use of undersmoothing to control the resulting bias. Second, in contrast to the unconditional situation \citep{szekely2004testing,gretton2012kernel}, kernel smoothing renders the U-statistics nondegenerate under the null hypotheses. However, when undersmoothing is used, the first-order projections in the Hoeffding decomposition of the U-statistics become asymptotically negligible under the null. Consequently, the asymptotic null distributions are determined by the second-order projections. Further details are provided in \ref{sec:local}-\ref{sec:global}. 

{
\renewcommand{\baselinestretch}{1.2}
\begin{table}[htbp]
\centering
\caption{Asymptotic behaviors in different scenarios, assuming the sample sizes $n_{1}\asymp n_{2}\asymp n$ and the bandwidths $h_{1}\asymp h_{2}\asymp h$.}
\begin{adjustbox}{width=\textwidth}
\begin{tabular}{cccccc}
\toprule
& \multicolumn{2}{c}{Global test} & & \multicolumn{2}{c}{Local test} \\
\cmidrule{2-3} \cmidrule{5-6}
& Null & Alternative & & Null & Alternative \\
\cmidrule{2-6}
Asymptotic distribution & Normal & Normal & & Weighted sum of chi-squares & Normal \\
Convergence rate & $n^{-1}h^{-p/2}$ & $n^{-1/2}$ & & $(nh^{p})^{-1}$ & $(nh^{p})^{-1/2}$ \\
\bottomrule
\end{tabular}
\end{adjustbox}
\label{tab:asy}
\end{table}
}

We now discuss related work and highlight several notable features of our framework. Although problems (\ref{eq:hypo}) and (\ref{eq:hypo2}) are fundamental in various modern applications, surprisingly, there are very few methods available to test the equality of two conditional distributions. In the nonparametric testing literature, existing methods mainly focus on testing the equality of conditional moments of $Y$ given $X$. Specifically, most studies aim at testing the equality of conditional means, also known as the comparison of regression curves, as seen in \citet{hall1990bootstrap,kulasekera1995comparison,dette1998nonparametric,lavergne2001equality,neumeyer2003nonparametric}, among others (see Section 7 in \citet{gonzalez2013updated} for a detailed review). Similarly, in the causal inference literature, hypothesis testing has largely been limited to conditional average treatment effect (see, e.g., \citet{crump2008nonparametric}). The methods we develop in this paper can detect general discrepancies between two conditional distributions, beyond specific moments. \citet{lee2009non} and \citet{chang2015nonparametric} proposed nonparametric tests for the null hypothesis of zero conditional distributional treatment effect. Their tests are built upon a Mann-Whitney statistic and cumulative distribution functions, respectively, and thus are only applicable for univariate $Y$. In our framework, $Y$ is allowed to take values in a general metric space, while $X$ can be multivariate. Very recently, utilizing techniques from conformal prediction, \citet{hu2024two} proposed a test for the global testing problem (\ref{eq:hypo}). However, their test involves (possibly unbalanced) sample splitting, and its performance relies on density ratio estimators of high quality (see Assumption 2(b) and related discussion therein). \citet{lee2024general} introduced two frameworks for the hypothesis (\ref{eq:hypo}). Their first framework converts a conditional independence test into a two-sample conditional distribution test, which effectively addresses an issue pointed out by an earlier version of this paper. Similar to \citet{hu2024two}, their second framework is based on density ratio estimation and sample splitting, which are not required in our work. Notably, the local testing problem (\ref{eq:hypo2}) is new and has not been addressed in the literature. Our framework can accommodate both global and local two-sample conditional distribution testing. \citet{chen2022paired} and \citet{chatterjee2024kernel} considered the paired-sample problem, which is very different from our two-sample setting. 

To end the introduction, we highlight the inherent challenges of two-sample conditional distribution testing in high-dimensional settings. When $Y$ is high-dimensional, even the simpler problem of unconditional two-sample testing lacks universally powerful methods. For example, distance and kernel-based approaches are known to suffer from the linearization effect as the dimension tends to infinity \citep{yan2023kernel}. Conversely, when $X$ is high-dimensional, even testing the equality of conditional means becomes difficult without imposing strong structural assumptions, such as a sparse linear model \citep{xia2018two}. Consequently, we do not consider high-dimensional settings in this work and instead leave them as an important direction for future research.

\section{Preliminaries}
This section provides an overview of the energy distance and maximum mean discrepancy, as well as their equivalence, which has been extensively discussed in \citet{sejdinovic2013equivalence}. 

For a nonempty set $\U$, a nonnegative function $\rho:\U\times\U\rightarrow[0,\infty)$ is called a semimetric on $\U$ if for any $u,u'\in\U$, it satisfies (i) $\rho(u,u')=0\iff u=u'$ and (ii) $\rho(u,u')=\rho(u',u)$. Then $(\U,\rho)$ is said to be a semimetric space. The semimetric space $(\U,\rho)$ is said to have negative type if for all $n\ge 2$, $u_{1},\ldots,u_{n}\in\U$ and $\alpha_{1},\ldots,\alpha_{n}\in\R$ with $\sum_{i=1}^{n}\alpha_{i}=0$, we have $\sum_{i,j=1}^{n}\alpha_{i}\alpha_{j}\rho(u_{i},u_{j})\le 0$. For $\theta>0$, define $\M_{\rho}^{\theta}(\U)=\{\mu\in\M(\U):\int\rho^{\theta}(u,u_{0})d\mu(u)<\infty$, for some $u_{0}\in\U\}$, where $\M(\U)$ denotes the set of all probability measures on $\U$. Suppose $P,Q\in\M_{\rho}^{1}(\U)$, we have $\int\rho d(P-Q)^{2}\le 0$ when $(\U,\rho)$ has negative type. We say that $(\U,\rho)$ has strong negative type if it has negative type and the equality holds only when $P=Q$. The (generalized) energy distance between $P,Q\in\M_{\rho}^{1}(\U)$ is defined as \citep{sejdinovic2013equivalence}
\[
\D_{\rho}(P,Q)=2\E\rho(U,V)-\E\rho(U,U')-\E\rho(V,V')=-\int\rho d(P-Q)^{2},
\]
where $U,U'\overset{\iid}{\sim}P$ and $V,V'\overset{\iid}{\sim}Q$. If $(\U,\rho)$ is of strong negative type, we have $\D_{\rho}(P,Q)\ge 0$ and the equality holds if and only if $P=Q$. Every separable Hilbert space is of strong negative type \citep{lyons2013distance}. In particular, Euclidean spaces are separable. 

Let $\H$ be a Hilbert space of real-valued functions defined on $\U$. A function $k:\U\times\U\rightarrow\R$ is a reproducing kernel of $\H$ if (i) $\forall u\in\U$, $k(\cdot,u)\in\H$ and (ii) $\forall u\in\U,\forall f\in\H$, $\langle f,k(\cdot,u)\rangle_{\H}=f(u)$, where $\langle\cdot,\cdot\rangle_{\H}$ is the inner product associated with $\H$. If $\H$ has a reproducing kernel, it is said to be a reproducing kernel Hilbert space (RKHS). According to Moore-Aronszajn theorem \citep{berlinet2011reproducing}, for every symmetric, positive definite function (henceforth kernel) $k$, there exists an associated RKHS $\H_{k}$ with reproducing kernel $k$. Define $\M_{k}^{1/2}(\U)=\{\mu\in\M(\U):\int \sqrt{k(u,u)}d\mu(u)<\infty\}$. The kernel mean embedding of $P\in\M_{k}^{1/2}(\U)$ into RKHS $\H_{k}$ is defined by the Bochner integral $\Pi_{k}(P)=\int k(\cdot,u)dP(u)\in\H_{k}$. The kernel $k$ is said to be characteristic if the mapping $\Pi_{k}$ is injective. Conditions under which kernels are characteristic have been studied by \citet{sriperumbudur2008injective,sriperumbudur2010hilbert}, and examples include the Gaussian and the Laplace kernel. The maximum mean discrepancy (MMD) between $P,Q\in\M_{k}^{1/2}(\U)$ is given by \citep{gretton2012kernel}
\[
\gamma_{k}(P,Q)=\|\Pi_{k}(P)-\Pi_{k}(Q)\|_{\H_k}.
\]
When the kernel $k$ is characteristic, we have $\gamma_{k}(P,Q)=0$ if and only if $P=Q$. Also, the following alternative representation of the squared MMD is useful
\[
\gamma_{k}^{2}(P,Q)=\E k(U,U')+\E k(V,V')-2\E k(U,V)=\int kd(P-Q)^{2},
\]
where $U,U'\overset{\iid}{\sim}P$ and $V,V'\overset{\iid}{\sim}Q$. 

Let $\rho$ be a semimetric of negative type on $\U$. For any $u_{0}\in\U$, the function $k(u,u')=\rho(u,u_{0})+\rho(u',u_{0})-\rho(u,u')$ is positive definite, and is said to be the distance-induced kernel induced by $\rho$ and centered at $u_{0}$. Correspondingly, for a kernel $k$ on $\U$, the function $\rho(u,u')=\frac{1}{2}\{k(u,u)+k(u',u')\}-k(u,u')$ defines a valid semimetric $\rho$ of negative type on $\U$, and we will say that $k$ generates $\rho$. It is clear that every distance-induced kernel $k$ induced by $\rho$, also generates $\rho$. Theorem 22 in \citet{sejdinovic2013equivalence} establishes the equivalence between the energy distance and MMD. Specifically, suppose $P,Q\in\M_{\rho}^{1}(\U)$ and let $k$ be any kernel that generates $\rho$, then $\D_{\rho}(P,Q)=\gamma_{k}^{2}(P,Q)$. 

\textit{Notation.} For $l=1,2$, let $f_{l}(y,x)$, $f_{l}(x)$ and $f_{l}(y\mid x)$ be the joint probability density function of $P^{(l)}$, the marginal probability density function of $P_{X}^{(l)}$ and the conditional probability density function of $P_{Y\mid X=x}^{(l)}$, respectively (assuming the existence of these densities). For two sequences of real numbers $\{a_{n}\}_{n=1}^{\infty}$ and $\{b_{n}\}_{n=1}^{\infty}$, we write $a_{n}\asymp b_{n}$ if and only if $a_{n}=O(b_{n})$ and $b_{n}=O(a_{n})$. The symbols $\overset{p}{\rightarrow}$ and $\overset{d}{\rightarrow}$ stand for convergence in probability and in distribution, respectively.

\section{Conditional energy distance and conditional maximum mean discrepancy}\label{sec:local}
Let $\rho$ be a semimetric on $\Y$. We define the conditional energy distance at $x\in\R^{p}$ as the energy distance between $P_{Y\mid X=x}^{(1)}$ and $P_{Y\mid X=x}^{(2)}$. 
\begin{definition}
    Assume $f_{l}(x)>0$ and $P_{Y\mid X=x}^{(l)}\in\M_{\rho}^{1}(\Y)$ for $l=1,2$. The conditional energy distance (CED) between $P^{(1)}$ and $P^{(2)}$ at $x\in\R^{p}$ is defined as
    \begin{align}\label{eq:ced}
        \begin{split}
            \D_{\rho}(x)&\coloneqq\D_{\rho}(P_{Y\mid X=x}^{(1)},P_{Y\mid X=x}^{(2)})=-\int_{\Y}\int_{\Y}\rho(y,y')d(P_{Y\mid X=x}^{(1)}-P_{Y\mid X=x}^{(2)})^{2}(y,y')\\
            &=2\E\{\rho(Y^{(1)},Y^{(2)})\mid X^{(1)}=x,X^{(2)}=x\}-\E\{\rho(Y^{(1)},Y^{(1)\prime})\mid X^{(1)}=x,X^{(1)\prime}=x\}\\
            &\quad-\E\{\rho(Y^{(2)},Y^{(2)\prime})\mid X^{(2)}=x,X^{(2)\prime}=x\},
        \end{split}
    \end{align}
    where $(Y^{(1)\prime},X^{(1)\prime})$ and $(Y^{(2)\prime},X^{(2)\prime})$ are \iid\ copies of $(Y^{(1)},X^{(1)})$ and $(Y^{(2)},X^{(2)})$, respectively. 
\end{definition}

When $\Y=\R^{q}$ and $\rho$ corresponds to the Euclidean distance, (\ref{eq:ced}) also has a nice interpretation in terms of conditional characteristic function. Specifically, for $l=1,2$ and $t\in\R^{q}$, the conditional characteristic function of $Y^{(l)}$ given $X^{(l)}=x$ is defined as $\varphi_{Y\mid X=x}^{(l)}(t)=\E\{\exp(\imath t^{\top}Y^{(l)})\mid X^{(l)}=x\}$, where $\imath=\sqrt{-1}$ is the imaginary unit. When $\rho(y,y')=\|y-y'\|$ with $\|\cdot\|$ being the Euclidean norm, we have
\[
\D_{\rho}(x)=\frac{1}{c_{q}}\int_{\R^{q}}\frac{|\varphi_{Y\mid X=x}^{(1)}(t)-\varphi_{Y\mid X=x}^{(2)}(t)|^{2}}{\|t\|^{q+1}}dt,
\]
where $c_{q}=\pi^{(q+1)/2}/\Gamma((q+1)/2)$ is a constant with $\Gamma(\cdot)$ being the gamma function. The proof of this fact follows a similar approach to that of Proposition 1 in \citet{szekely2017energy}, and is therefore omitted for brevity. 

As the counterpart to the distance-based metric (\ref{eq:ced}), we now introduce a kernel-based metric to measure the discrepancy between the two conditional distributions. Let $\H_{k}$ be an RKHS associated with a kernel $k$ on $\Y$. 
\begin{definition}
    Assume $f_{l}(x)>0$ and $P_{Y\mid X=x}^{(l)}\in\M_{k}^{1/2}(\Y)$ for $l=1,2$. The conditional maximum mean discrepancy (CMMD) between $P^{(1)}$ and $P^{(2)}$ at $x\in\R^{p}$ is defined as the square root of
    \begin{align}\label{eq:cmmd}
        \begin{split}
            \gamma_{k}^{2}(x)&\coloneqq\gamma_{k}^{2}(P_{Y\mid X=x}^{(1)},P_{Y\mid X=x}^{(2)})=\|\Pi_{k}(P_{Y\mid X=x}^{(1)})-\Pi_{k}(P_{Y\mid X=x}^{(2)})\|_{\H_{k}}^{2}\\
            &=\int_{\Y}\int_{\Y}k(y,y')d(P_{Y\mid X=x}^{(1)}-P_{Y\mid X=x}^{(2)})^{2}(y,y')\\
            &=\E\{k(Y^{(1)},Y^{(1)\prime})\mid X^{(1)}=x,X^{(1)\prime}=x\}+\E\{k(Y^{(2)},Y^{(2)\prime})\mid X^{(2)}=x,X^{(2)\prime}=x\}\\
            &\quad-2\E\{k(Y^{(1)},Y^{(2)})\mid X^{(1)}=x,X^{(2)}=x\},
        \end{split}
    \end{align}
    where $(Y^{(1)\prime},X^{(1)\prime})$ and $(Y^{(2)\prime},X^{(2)\prime})$ are \iid\ copies of $(Y^{(1)},X^{(1)})$ and $(Y^{(2)},X^{(2)})$, respectively. 
\end{definition}

While conducting this research, we came across \citet{park2020measure}, in which the authors present the CMMD in a slightly different form. However, their work does not address two-sample conditional distribution testing, and their estimation strategy differs entirely from ours. Specifically, the computational complexity of their estimator is $O((n_{1}+n_{2})^{3})$, whereas ours has a lower complexity of $O((n_{1}+n_{2})^{2})$ (see Remark \ref{remark:time}). Moreover, \citet{park2020measure} does not provide the convergence rate or the asymptotic distribution of their estimator. In contrast, we derive both the exact convergence rate and the asymptotic distribution of our statistic under both the null and alternative hypotheses. Additionally, we develop a unified framework that encompasses both distance and kernel-based measures. It is also important to note that the CMMD is a metric indexed by $x\in\R^{p}$, and the single measure (\ref{eq:i}) introduced in Section \ref{sec:global}, which integrates the CMMD with a weight function, is not considered by \citet{park2020measure}.

The CED and CMMD both serve to characterize the homogeneity of two conditional distributions. Besides, we can demonstrate the equivalence between the CED and CMMD, which holds for general semimetric of negative type. Similar results have been established for distance and kernel-based measures used in two-sample/independence testing \citep{sejdinovic2013equivalence} and conditional independence testing \citep{sheng2023distance}. Consequently, we will focus on the CED for the remainder of this paper.
\begin{theorem}\label{thm1}~
    \begin{enumerate}
        \item When the semimetric $\rho$ is of strong negative type, for any $x\in\R^{p}$ such that $f_{l}(x)>0$ and $P_{Y\mid X=x}^{(l)}\in\M_{\rho}^{1}(\Y)\ (l=1,2)$, we have $\D_{\rho}(x)\ge 0$, and $\D_{\rho}(x)=0$ if and only if $H_{0}$ in (\ref{eq:hypo2}) holds. 
        \item When the kernel $k$ is characteristic, for any $x\in\R^{p}$ such that $f_{l}(x)>0$ and $P_{Y\mid X=x}^{(l)}\in\M_{k}^{1/2}(\Y)\ (l=1,2)$, we have $\gamma_{k}(x)=0$ if and only if $H_{0}$ in (\ref{eq:hypo2}) holds. 
        \item Let $\rho$ be a semimetric of negative type on $\Y$. Suppose $f_{l}(x)>0$ and $P_{Y\mid X=x}^{(l)}\in\M_{\rho}^{1}(\Y)$ for $l=1,2$. If $k$ is a kernel that generates $\rho$, then
        \[
        \D_{\rho}(x)=\gamma_{k}^{2}(x). 
        \]
    \end{enumerate}
\end{theorem}


We now present an estimator for the CED. Let $x(s)$ denote the $s$th component of $x\in\R^{p}$ for $s=1,\ldots,p$. For $l=1,2$, define $G_{h_{l}}:\R^{p}\rightarrow\R$ as
\[
G_{h_{l}}(x)=\frac{1}{h_{l}^{p}}\prod_{s=1}^{p}g\bigg(\frac{x(s)}{h_{l}}\bigg),\quad x\in\R^{p},
\]
where $g:\R\rightarrow\R$ is a univariate kernel function satisfying Assumption \ref{assum1} below, and $h_{1},h_{2}\in\R$ are the bandwidth parameters. For ease of presentation, here we use the same bandwidth $h_l$ for each component of $X^{(l)}$. In practice, one should always allow for varying bandwidths across each component of $X^{(l)}$. 
Motivated by the representation of CED in terms of the conditional moments given in (\ref{eq:ced}), we propose the following estimator for the CED: 
\begin{align*}
    \widehat{\D}_{\rho}(x)&=2\bigg\{\sum_{i_{1}\neq i_{2}}G_{h_{1}}(X_{i_{1}}^{(1)}-x)G_{h_{1}}(X_{i_{2}}^{(1)}-x)\sum_{j_{1}\neq j_{2}}G_{h_{2}}(X_{j_{1}}^{(2)}-x)G_{h_{2}}(X_{j_{2}}^{(2)}-x)\bigg\}^{-1}\\
    &\quad\times\sum_{i=1}^{n_{1}}\sum_{j=1}^{n_{2}}\rho(Y_{i}^{(1)},Y_{j}^{(2)})G_{h_{1}}(X_{i}^{(1)}-x)G_{h_{2}}(X_{j}^{(2)}-x)\sum_{i'\neq i}G_{h_{1}}(X_{i'}^{(1)}-x)\sum_{j'\neq j}G_{h_{2}}(X_{j'}^{(2)}-x)\\
    &\quad-\bigg\{\sum_{i_{1}\neq i_{2}}G_{h_{1}}(X_{i_{1}}^{(1)}-x)G_{h_{1}}(X_{i_{2}}^{(1)}-x)\bigg\}^{-1}\sum_{i_{1}\neq i_{2}}\rho(Y_{i_{1}}^{(1)},Y_{i_{2}}^{(1)})G_{h_{1}}(X_{i_{1}}^{(1)}-x)G_{h_{1}}(X_{i_{2}}^{(1)}-x)\\
    &\quad-\bigg\{\sum_{j_{1}\neq j_{2}}G_{h_{2}}(X_{j_{1}}^{(2)}-x)G_{h_{2}}(X_{j_{2}}^{(2)}-x)\bigg\}^{-1}\sum_{j_{1}\neq j_{2}}\rho(Y_{j_{1}}^{(2)},Y_{j_{2}}^{(2)})G_{h_{2}}(X_{j_{1}}^{(2)}-x)G_{h_{2}}(X_{j_{2}}^{(2)}-x),
\end{align*}
which we call the sample conditional energy distance. Our estimation strategy, which combines U-statistics and kernel smoothing, has been employed in previous studies such as \citet{wang2015conditional,ke2020expected}. It first appeared in the so-called conditional U-statistics \citep{stute1991conditional}, which generalize the Nadaraya-Watson estimate of a regression function in the same way as Hoeffding's U-statistic is a generalization of the sample mean. Instead of V-statistics, we adopt U-statistics in this paper to avoid unnecessary bias; see the proof of Theorem \ref{thm2} for the explicit U-statistic form. 

\begin{remark}\label{remark:y}
    The definitions of the CED and CMMD, as well as the estimation procedures introduced, are applicable to a general metric space $\Y$. This implies that our framework is capable of handling cases where $Y$ is a non-Euclidean object, such as a curve, image, or even topological data encountered in various applications.
\end{remark}

In what follows, we present the consistency of the sample CED in Theorem \ref{thm2}, while Theorems \ref{thm3} and \ref{thm4} provide its asymptotic distributions under $H_{a}$ and $H_{0}$ in (\ref{eq:hypo2}), respectively. We shall make the following assumptions throughout the analysis. 
\begin{assumption}\label{assum1}
    The univariate smoothing kernel $g$ is of order $\nu$ in the sense that $\int g(u)du=1$, $\int u^{r}g(u)du=0$ for $r=1,\cdots,\nu-1$, and $\int u^{\nu}g(u)du\neq 0$. Also, $g$ is bounded with $\int g^{2}(u)du<\infty$ and $\int |u|^{\nu}|g(u)|du<\infty$. 
\end{assumption}
\begin{assumption}\label{assum2}
    For $l=1,2$, $h_{l}\rightarrow 0$ and $n_{l}h_{l}^{p}\rightarrow\infty$, as $n_{l}\rightarrow\infty$. 
\end{assumption}
\begin{assumption}\label{assum3}
    The marginal densities $f_{1}(x)$ and $f_{2}(x)$ are $\nu$-times continuously differentiable. 
\end{assumption}
\begin{assumption}\label{assum4}
    $\E\{\rho(Y^{(1)},Y^{(2)})\mid X^{(1)}=x_{1},X^{(2)}=x_{2}\}$, $\E\{\rho(Y^{(1)},Y^{(1)\prime})\mid X^{(1)}=x_{1},X^{(1)\prime}=x_{2}\}$, and $\E\{\rho(Y^{(2)},Y^{(2)\prime})\mid X^{(2)}=x_{1},X^{(2)\prime}=x_{2}\}$ are $\nu$-times continuously differentiable with respect to $(x_{1}^{\top},x_{2}^{\top})^{\top}$. 
\end{assumption}

Assumptions \ref{assum1}-\ref{assum3} are standard in the literature (see, e.g., Section 1.11 in \citet{li2007nonparametric}), allowing for multivariate $X$ and higher-order kernels. Assumption \ref{assum4} can be seen as the counterpart of the usual smoothness condition imposed on the regression function in classical kernel regression. A similar assumption is used in \citet{stute1991conditional}. 

\begin{theorem}\label{thm2}
    Suppose $f_{l}(x)>0$ and $P_{Y\mid X=x}^{(l)}\in\M_{\rho}^{2}(\Y)$ for $l=1,2$. Under Assumptions \ref{assum1}-\ref{assum4}, we have
    \[
    \widehat{\D}_{\rho}(x)\convp\D_{\rho}(x),\quad\text{as }n_{1},n_{2}\rightarrow\infty.
    \]
\end{theorem}

\begin{theorem}\label{thm3}
    Suppose $f_{l}(x)>0$ and $P_{Y\mid X=x}^{(l)}\in\M_{\rho}^{2+\delta}(\Y)$ for $l=1,2$ and some $\delta>0$. Under $H_{a}$ in (\ref{eq:hypo2}), Assumptions \ref{assum1}-\ref{assum4} and $n_{l}^{1/2}h_{l}^{p/2+\nu}\rightarrow 0$ as $n_{l}\rightarrow\infty$ for $l=1,2$, we have
    \[
    \frac{\widehat{\D}_{\rho}(x)-\D_{\rho}(x)}{\frac{2}{\{f_{1}(x)f_{2}(x)\}^{2}}\sqrt{\frac{\xi_{1}^{2}}{n_{1}}+\frac{\xi_{2}^{2}}{n_{2}}}}\convd\N(0,1),\quad\text{as }n_{1},n_{2}\rightarrow\infty,
    \]
    where $\xi_{1}^{2}$ and $\xi_{2}^{2}$ are defined in equations (\ref{eq-xi1})-(\ref{eq-xi2}) of Appendix \ref{app:c}.
\end{theorem}

As shown in the proof, we have $\xi_{1}^{2}\asymp h_{1}^{-p}$ and $\xi_{2}^{2}\asymp h_{2}^{-p}$. Thus, under $H_{a}$ in (\ref{eq:hypo2}), $\widehat{\D}_{\rho}(x)-\D_{\rho}(x)=O_{p}((n_{1}h_{1}^{p})^{-1/2}+(n_{2}h_{2}^{p})^{-1/2})$, which is the same convergence rate as that of classical kernel regression \citep{ullah1999nonparametric} and conditional U-statistics \citep{stute1991conditional}. The condition $n_{l}^{1/2}h_{l}^{p/2+\nu}\rightarrow 0$ as $n_{l}\rightarrow\infty$ in Theorem \ref{thm3} requires undersmoothing in order to make the estimator asymptotically unbiased. Undersmoothing, coupled with the use of a higher-order kernel, is commonly employed in nonparametric and semiparametric estimation to reduce the bias of estimators involving kernel smoothing \citep{li2007nonparametric}. The same condition with $p=1$ and $\nu=2$ is used in Section 3.4 of \citet{ullah1999nonparametric} for the central limit theorem (CLT) of classical kernel regression and in \citet{stute1991conditional} for the CLT of conditional U-statistics.

\begin{theorem}\label{thm4}
    Define the double centered semimetric
    \[
    \rhotilde(y,y')=\rho(y,y')-\E\{\rho(y,Y')\mid X'=x\}-\E\{\rho(Y,y')\mid X=x\}+\E\{\rho(Y,Y')\mid X=x,X'=x\},
    \]
    where $Y\mid X=x,Y'\mid X'=x\overset{\iid}{\sim}P_{Y\mid X=x}\coloneqq P_{Y\mid X=x}^{(1)}=P_{Y\mid X=x}^{(2)}$ under $H_{0}$ in (\ref{eq:hypo2}). Suppose $f_{l}(x)>0$ and $P_{Y\mid X=x}\in\M_{\rho}^{2+\delta}(\Y)$ for $l=1,2$ and some $\delta>0$. Under $H_{0}$ in (\ref{eq:hypo2}), Assumptions \ref{assum1}-\ref{assum4}, $n_{1}h_{1}^{p}/(n_{1}h_{1}^{p}+n_{2}h_{2}^{p})\rightarrow\tau\in(0,1)$, and $n_{l}h_{l}^{p+\nu}\rightarrow 0$ as $n_{l}\rightarrow\infty$ for $l=1,2$, we have
    \[
    (n_{1}h_{1}^{p}+n_{2}h_{2}^{p})\widehat{\D}_{\rho}(x)\convd-\sum_{r=1}^{\infty}\lambda_{r}\bigg\{\bigg(\frac{1}{\sqrt{\tau}}\zeta_{r}-\frac{1}{\sqrt{1-\tau}}\eta_{r}\bigg)^{2}-\frac{(1-\tau)\sigma_{1}^{2}+\tau\sigma_{2}^{2}}{\tau(1-\tau)}\bigg\},
    \]
    where $\{\zeta_r\}_{r=1}^{\infty}$ and $\{\eta_r\}_{r=1}^{\infty}$ are two sequences of independent normal random variables with
    \begin{align*}
        \zeta_{r}&\overset{\iid}{\sim}\N(0,\sigma_{1}^{2}),\quad\sigma_{1}^{2}=\frac{1}{f_{1}(x)}\bigg\{\int g^{2}(u)du\bigg\}^{p},\\
        \eta_{r}&\overset{\iid}{\sim}\N(0,\sigma_{2}^{2}),\quad\sigma_{2}^{2}=\frac{1}{f_{2}(x)}\bigg\{\int g^{2}(u)du\bigg\}^{p},
    \end{align*}
    and $\{\lambda_{r}\}_{r=1}^{\infty}$ are the eigenvalues of $\rhotilde$ with respect to the conditional distribution $P_{Y\mid X=x}$, i.e., 
    \[
    \E\{\rhotilde(y,Y')\phi_{r}(Y')\mid X'=x\}=\lambda_{r}\phi_{r}(y).
    \]
\end{theorem}

As demonstrated in Theorem \ref{thm4}, the sample CED converges at a faster rate under $H_{0}$ in (\ref{eq:hypo2}): $\widehat{\D}_{\rho}(x)=O_{p}((n_{1}h_{1}^{p})^{-1}+(n_{2}h_{2}^{p})^{-1})$, compared to that under $H_{a}$ in (\ref{eq:hypo2}). This renders the undersmoothing condition in Theorem \ref{thm4} more stringent than that in Theorem \ref{thm3}. Compared with the unconditional scenario \citep{gretton2012kernel}, our results differ in several aspects. First, the convergence rates are distinct. Second, our spectral decomposition is performed with respect to the conditional distribution. Third, in contrast to \iid\ standard normal random variables in \citet{gretton2012kernel}, $\{\zeta_r\}_{r=1}^{\infty}$ and $\{\eta_r\}_{r=1}^{\infty}$ have different variances when $f_{1}(x)\neq f_{2}(x)$.

\section{Integrated conditional energy distance}\label{sec:global}
The CED is indexed by $x\in\R^{p}$, and thus is not a single number. We can obtain a single measure by integrating the CED with some weight function $\omega:\R^{p}\rightarrow\R$, i.e., 
\[
\I_{\rho}^{\omega}\coloneqq\int_{\R^{p}}\D_{\rho}(x)\omega(x)dx, 
\]
which we call the integrated conditional energy distance (ICED). It follows directly from Theorem \ref{thm1} that $\I_{\rho}^{\omega}$ fully characterizes the homogeneity of two conditional distributions. 
\begin{corollary}
    When $\omega$ is a nonnegative function with the same support as $P^{(1)}_{X}$ (or equivalently $P^{(2)}_{X}$) and the semimetric $\rho$ is of strong negative type, we have $\I_{\rho}^{\omega}=0$ if and only if $H_{0}$ in (\ref{eq:hypo}) holds.  
\end{corollary}

Motivated by \citet{su2007consistent,wang2015conditional,ke2020expected}, we consider the weight function $\omega(x)=\{f_{1}(x)f_{2}(x)\}^{2}$ and define
\begin{equation}\label{eq:i}
    \I_{\rho}\coloneqq\int_{\R^{p}}\D_{\rho}(x)\{f_{1}(x)f_{2}(x)\}^{2}dx.
\end{equation}
This particular choice of weight function circumvents the random denominator issue. Otherwise, to deal with the density near zero and mitigate significant bias, additional stringent assumptions or trimming schemes may need to be used.  
For example, \citet{yin2020new} assumes that the density functions $f_{1}(x)$ and $f_{2}(x)$ are bounded away from zero, which can be quite restrictive. 

We employ the same estimation strategy as in Section \ref{sec:local} and propose the following estimator for the ICED, which is a U-statistic by symmetrization (see the proof of Theorem \ref{thm5}):
\begin{align*}
    \widehat{\I}_{\rho}&=\frac{1}{n_{1}(n_{1}-1)}\frac{1}{n_{2}(n_{2}-1)}\bigg[\sum_{i=1}^{n_{1}}\sum_{j=1}^{n_{2}}\rho(Y_{i}^{(1)},Y_{j}^{(2)})\Big\{G_{h_{1}}(X_{i}^{(1)}-X_{j}^{(2)})+G_{h_{2}}(X_{j}^{(2)}-X_{i}^{(1)})\Big\}\\
    &\quad\times\sum_{i'\neq i}G_{h_{1}}(X_{i'}^{(1)}-X_{j}^{(2)})\sum_{j'\neq j}G_{h_{2}}(X_{j'}^{(2)}-X_{i}^{(1)})\\
    &\quad-\sum_{i_{1}\neq i_{2}}\rho(Y_{i_{1}}^{(1)},Y_{i_{2}}^{(1)})G_{h_{1}}(X_{i_{1}}^{(1)}-X_{i_{2}}^{(1)})\sum_{j_{1}\neq j_{2}}G_{h_{2}}(X_{j_{1}}^{(2)}-X_{i_{1}}^{(1)})G_{h_{2}}(X_{j_{2}}^{(2)}-X_{i_{1}}^{(1)})\\
    &\quad-\sum_{j_{1}\neq j_{2}}\rho(Y_{j_{1}}^{(2)},Y_{j_{2}}^{(2)})G_{h_{2}}(X_{j_{1}}^{(2)}-X_{j_{2}}^{(2)})\sum_{i_{1}\neq i_{2}}G_{h_{1}}(X_{i_{1}}^{(1)}-X_{j_{1}}^{(2)})G_{h_{1}}(X_{i_{2}}^{(1)}-X_{j_{1}}^{(2)})\bigg].
\end{align*}
We call $\widehat{\I}_{\rho}$ the sample integrated conditional energy distance. 

\begin{remark}\label{remark:time}
    The computational complexity of both the sample CED $\widehat{\D}_{\rho}(x)$ and sample ICED $\widehat{\I}_{\rho}$ is $O((n_{1}+n_{2})^{2})$, which is the same as that of the unconditional situation \citep{szekely2004testing,gretton2012kernel}. The details are provided in Appendix \ref{app:a}.
\end{remark}

Theorem \ref{thm5} asserts that $\widehat{\I}_{\rho}$ is a consistent estimator of $\I_{\rho}$. By examining the Hoeffding decomposition of $\widehat{\I}_{\rho}$, Theorem \ref{thm6} describes its asymptotic distribution under the alternative hypothesis $H_{a}$ in (\ref{eq:hypo}), while Theorem \ref{thm7} characterizes its asymptotic null distribution under $H_{0}$ in (\ref{eq:hypo}). To get these results, we need to strengthen Assumptions \ref{assum3}-\ref{assum4} in the following manner.
\begin{assumption}\label{assum5}
    In addition to Assumption \ref{assum3}, we require the derivatives of $f_{1}(x)$ and $f_{2}(x)$ to be bounded uniformly over $x$. 
\end{assumption}
\begin{assumption}\label{assum6}
    In addition to Assumption \ref{assum4}, we require the derivatives of the bivariate functions defined therein to belong to the set $\{f:\R^{p}\times\R^{p}\rightarrow\R,\ \E_{P_{X}^{(1)}}|f(X,X)|<\infty,\ \E_{P_{X}^{(2)}}|f(X,X)|<\infty\}$. 
\end{assumption}

Assumption \ref{assum5} is standard, as seen in previous studies such as \citet{lee2009non,wang2015conditional}. Assumption \ref{assum6} is mild in that it does not require bounded derivatives of the conditional densities or moments. \citet{wang2015conditional} assumes that the conditional densities have bounded derivatives. However, under this assumption, their proof implicitly requires the distance function $\|y_{1}-y_{2}\|$ to be integrable with respect to the Lebesgue measure (not some probability measure), which fails to hold unless the support is bounded. A careful examination of the arguments in \citet{wang2015conditional} and \citet{ke2020expected} suggests that conditions analogous to our Assumption \ref{assum6} are necessary.

\begin{theorem}\label{thm5}
    Suppose $P_{Y}^{(l)}\in\M_{\rho}^{2}(\Y)$ for $l=1,2$. Under Assumptions \ref{assum1}-\ref{assum2} and \ref{assum5}-\ref{assum6}, we have
    \[
    \widehat{\I}_{\rho}\convp\I_{\rho},\quad\text{as }n_{1},n_{2}\rightarrow\infty.
    \]
\end{theorem}

\begin{theorem}\label{thm6}
    Suppose $P_{Y}^{(l)}\in\M_{\rho}^{2+\delta}(\Y)$ for $l=1,2$ and some $\delta>0$. Under $H_{a}$ in (\ref{eq:hypo}), Assumptions \ref{assum1}-\ref{assum2}, \ref{assum5}-\ref{assum6} and $n_{l}^{1/2}h_{l}^{\nu}\rightarrow 0$ as $n_{l}\rightarrow\infty$ for $l=1,2$, we have
    \[
    \frac{\widehat{\I}_{\rho}-\I_{\rho}}{2\sqrt{\frac{\delta_{1,0}^{2}}{n_{1}}+\frac{\delta_{0,1}^{2}}{n_{2}}}}\convd\N(0,1),\quad\text{as }n_{1},n_{2}\rightarrow\infty,
    \]
    where $\delta_{1,0}^{2}$ and $\delta_{0,1}^{2}$ are defined in equations (\ref{eq:delta10})-(\ref{eq:delta01}) of Appendix \ref{app:c}.
\end{theorem}

As shown in the proof, we have $\delta_{1,0}^{2}\asymp 1$ and $\delta_{0,1}^{2}\asymp 1$. Therefore, under $H_{a}$ in (\ref{eq:hypo}), $\widehat{\I}_{\rho}-\I_{\rho}=O_{p}(n_{1}^{-1/2}+n_{2}^{-1/2})$. Again, undersmoothing is necessary, and the same condition has been used in \citet{lee2009non}. Moreover, we emphasize that undersmoothing is not needed for establishing the consistency of $\widehat{\D}_{\rho}(x)$ and $\widehat{\I}_{\rho}$. 

\begin{theorem}\label{thm7}
    Suppose $P_{Y}^{(l)}\in\M_{\rho}^{2+\delta}(\Y)$ for $l=1,2$ and some $\delta>0$. Under $H_{0}$ in (\ref{eq:hypo}), Assumptions \ref{assum1}-\ref{assum2}, \ref{assum5}-\ref{assum6} and $n_{l}h_{l}^{p/2+\nu}\rightarrow 0$ as $n_{l}\rightarrow\infty$ for $l=1,2$, we have
    \[
    \frac{\widehat{\I}_{\rho}}{\sqrt{\frac{2\delta_{2,0}^{2}}{n_{1}(n_{1}-1)}+\frac{16\delta_{1,1}^{2}}{n_{1}n_{2}}+\frac{2\delta_{0,2}^{2}}{n_{2}(n_{2}-1)}}}\convd\N(0,1),\quad\text{as }n_{1},n_{2}\rightarrow\infty,
    \]
    where $\delta_{2,0}^{2}$, $\delta_{1,1}^{2}$ and $\delta_{0,2}^{2}$ are defined in equations (\ref{eq:delta20})-(\ref{eq:delta02}) of Appendix \ref{app:c}. 
\end{theorem}

As shown in the proof, we have $\delta_{2,0}^{2}\asymp h_{1}^{-p}$, $\delta_{1,1}^{2}\asymp h_{1}^{-p}+h_{2}^{-p}$ and $\delta_{0,2}^{2}\asymp h_{2}^{-p}$. Hence, the convergence rate of the sample ICED under $H_{0}$ in (\ref{eq:hypo}) is $\widehat{\I}_{\rho}=O_{p}(n_{1}^{-1}h_{1}^{-p/2}+n_{2}^{-1}h_{2}^{-p/2})$. Because of the integration of $x$ over $\R^{p}$, the convergence rate of the global statistic $\widehat{\I}_{\rho}$ is faster than that of the local statistic $\widehat{\D}_{\rho}(x)$ under both the null and alternative hypotheses. 
\begin{remark}
    Instead of a weighted sum of chi-squares, the global statistic $\widehat{\I}_{\rho}$ has a limiting normal distribution under the null. This distinct behavior follows from the martingale CLT (see, e.g., Corollary 3.1 of \citet{hall2014martingale}). The asymptotic distributions of both $\widehat{\D}_{\rho}(x)$ and $\widehat{\I}_{\rho}$ are governed by the second-order projections in their respective Hoeffding decompositions. However, the local statistic $\widehat{\D}_{\rho}(x)$ does not satisfy a key condition required for the application of the martingale CLT, and thus behaves similar to its unconditional counterparts \citep{szekely2004testing,gretton2012kernel}. In contrast, the global statistic $\widehat{\I}_{\rho}$ does meet the necessary conditions for the martingale CLT to hold. See also \citet{hall1984central} for related insights. 
\end{remark}

Theorems \ref{thm3}-\ref{thm4} and \ref{thm6}-\ref{thm7} demonstrate that kernel smoothing has two key effects on the asymptotic behavior of distance and kernel-based measures:
\begin{itemize}
    \item Kernel smoothing introduces biases of order $O(h_{1}^{\nu}+h_{2}^{\nu})$. To ensure valid inference, undersmoothing or other bias reduction techniques are necessary to achieve asymptotic unbiasedness. For the local statistic $\widehat{\D}_{\rho}(x)$, Theorem \ref{thm4} and Assumption \ref{assum2} require the bandwidths to satisfy $h_{l}\asymp n_{l}^{\kappa_{1}}$, where $\kappa_{1}\in(-\frac{1}{p},-\frac{1}{p+\nu})$. Notably, there is no restriction on the relationship between the kernel order $\nu$ and the dimension $p$ in this case. In contrast, for the global statistic $\widehat{\I}_{\rho}$, Theorem \ref{thm7} and Assumption \ref{assum2} impose the condition $h_{l}\asymp n_{l}^{\kappa_{2}}$, with $\kappa_{2}\in(-\frac{1}{p},-\frac{1}{p/2+\nu})$. This requires $\nu>p/2$, indicating that higher-order kernels are necessary as the dimension $p$ increases. In particular, the use of a second-order kernel, as in \citet{wang2015conditional}, becomes inadequate when $p\ge 4$ (see their Theorem 7). 

    \item As shown in the proof, due to kernel smoothing, the U-statistic estimators $\widehat{\D}_{\rho}(x)$ and $\widehat{\I}_{\rho}$ are nondegenerate under the null hypotheses. This stands in stark contrast to the unconditional setting, where the U-statistic estimators for energy distance and maximum mean discrepancy are degenerate under the null \citep{szekely2004testing,gretton2012kernel}. Nonetheless, when undersmoothing is used, a more in-depth analysis reveals that the first-order projections in their Hoeffding decompositions are asymptotically negligible under the null, relative to the second-order projections. As a result, the asymptotic null distributions are governed by the second-order projections. To be specific, we have $\xi_{1}^{2}=O(h_{1}^{-p}(h_{1}^{2\nu}+h_{2}^{2\nu}))$ under $H_{0}$ in (\ref{eq:hypo2}), instead of $\xi_{1}^{2}\asymp h_{1}^{-p}$ under $H_{a}$ in (\ref{eq:hypo2}). Similarly, we have $\delta_{1,0}^{2}=O(h_{1}^{2\nu}+h_{2}^{2\nu})$ under $H_{0}$ in (\ref{eq:hypo}), instead of $\delta_{1,0}^{2}\asymp 1$ under $H_{a}$ in (\ref{eq:hypo}). 
\end{itemize}

\begin{remark}\label{remark:errata}
    Our work utilizes techniques similar to those used by \citet{wang2015conditional} and \citet{ke2020expected}, including U-statistics and kernel smoothing. However, we believe that there exists a gap in the derivations of the asymptotic null distribution in both studies. Specifically, Theorem 7 of \citet{wang2015conditional} and Theorem 9 of \citet{ke2020expected} both rely on Lemma B.4 of \citet{fan1996consistent}, which establishes the CLT for degenerate U-statistics under certain conditions. To show that a U-statistic is degenerate, it must be demonstrated that $\mathcal{P}_{n1}(W_{1})=0$ under the null hypothesis, as per the notation used in their proofs. However, \citet{wang2015conditional} merely showed that $\E\{\mathcal{P}_{n1}(W_{1})\}=O(h^{2})$, while \citet{ke2020expected} incorrectly claimed that $\E\{\mathcal{P}_{n1}(W_{1})\}=0$. As a result, both studies overlook the crucial undersmoothing step necessary for valid inference. 
\end{remark}

\section{Applications to global and local two-sample conditional distribution testing}\label{sec:applications}
Theorems \ref{thm5}-\ref{thm7} collectively suggest that $\widehat{\I}_{\rho}$ is a suitable test statistic for testing (\ref{eq:hypo}). Nevertheless, it is impractical to use Theorem \ref{thm7} to compute the $p$-value since it is arduous to estimate $\delta_{2,0}^{2}$, $\delta_{1,1}^{2}$ and $\delta_{0,2}^{2}$. Moreover, it is widely acknowledged that a nonparametric test that relies on asymptotic normal approximation may perform poorly in finite samples \citep{su2007consistent}. Consequently, we resort to the local bootstrap proposed by \citet{paparoditis2000local}. This approach has been widely employed in other works involving conditional distributions; see, e.g.,  \citet{su2008nonparametric,huang2010testing,bouezmarni2012nonparametric,su2013nonparametric,taamouti2014nonparametric,wang2015conditional}. One can follow the aforementioned references to verify the asymptotic validity of this bootstrap method in our framework. 
Define
\begin{equation}
    \widehat{P}_{Y\mid X=x}=\bigg\{\sum_{l=1}^{2}\sum_{i=1}^{n_l}G_{b_l}(X_i^{(l)}-x)\bigg\}^{-1}\sum_{l=1}^{2}\sum_{i=1}^{n_l} G_{b_l}(X_i^{(l)}-x) \delta_{Y_i^{(l)}}, \label{eq:P_local_boot}
\end{equation}
where $\delta_{y}$ denotes a point mass at $y\in\Y$, and $b_1, b_2\in\R$ are the bandwidth parameters for local bootstrap. 
Essentially, $\widehat{P}_{Y \mid X = x}$ is a discrete distribution that assigns the probability $G_{b_l}(X_i^{(l)}-x) / \sum_{l=1}^{2}\sum_{i=1}^{n_l}G_{b_l}(X_i^{(l)}-x)$ to the observation $Y_i^{(l)}$. Then the following steps outline the procedure for the global two-sample conditional distribution test:
\begin{enumerate}[(i)]
    \item\label{enum:local_boot_step1} Calculate the global test statistic, i.e., the sample ICED, $\widehat{\I}_{\rho}$.
    \item\label{enum:local_boot_step2} For $l=1,2$ and $i=1,\ldots,n_{l}$, draw $\widehat{Y}_{i}^{(l)}$ from $\widehat{P}_{Y \mid X = X_i^{(l)}}$. Calculate the global test statistic using the local bootstrap samples $\{(\widehat{Y}_{i}^{(l)},X_{i}^{(l)})\}_{i=1}^{n_{l}}\ (l=1,2)$, which retains the marginal distributions of $X$ but imposes the null restriction (i.e., the same conditional distribution of $Y$ given $X$). 
    \item\label{enum:local_boot_step3} Repeat step \ref{enum:local_boot_step2} $B$ times, and collect $\{\widehat{\I}_{\rho}^{(t)}\}_{t=1}^{B}$. Then, the bootstrap-based $p$-value of the global test is given by
    \[
    p\text{-value}=\frac{\sum_{t=1}^{B}\mathbbm{1}\{\widehat{\I}_{\rho}^{(t)}>\widehat{\I}_{\rho}\}+1}{B+1}.
    \]
\end{enumerate}

Theorems \ref{thm2}-\ref{thm4} support the use of $\widehat{\D}_{\rho}(x)$ as the test statistic for testing (\ref{eq:hypo2}). 
Unfortunately, the asymptotic null distribution of $\widehat{\D}_{\rho}(x)$ in Theorem \ref{thm4} is not pivotal and involves infinite nuisance parameters.
To tackle this issue, we once again utilize the local bootstrap method to calculate the $p$-value, replacing the $\widehat{\I}_{\rho}$ in the above steps \ref{enum:local_boot_step1}-\ref{enum:local_boot_step3} with $\widehat{\D}_{\rho}(x)$. 

For the smoothing bandwidth parameters $h_l$ in the test statistics, we introduce a completely data-driven selection approach, which works well in our numerical studies.
Specifically, for $l = 1, 2$, we determine $h_l$ by minimizing the cross-validation loss function
\[
\operatorname{CV}(h_l(1),\ldots,h_{l}(p)) = \sum_{i=1}^{n_l} \| k(Y_i^{(l)}, \cdot) - \widehat{\Pi}_{k}^{-i}(P_{Y \mid X=X_{i}^{(l)}}^{(l)}) \|_{\mathcal{H}_k}^{2},
\]
where $\widehat{\Pi}_{k}^{-i}(P_{Y \mid X=x}^{(l)}) = \{\sum_{i'\ne i} G_{h_l}(X_{i'}^{(l)} - x)\}^{-1} \sum_{i'\ne i} G_{h_l}(X_{i'}^{(l)} - x) k(Y_{i'}^{(l)}, \cdot)$ is the leave-one-out Nadaraya-Watson estimator of the conditional kernel mean embedding $\Pi_{k}(P_{Y \mid X=x}^{(l)})$. The minimization problem is solved using the limited-memory quasi-Newton method with box constraints \citep{byrd1995limited}, implemented via the \texttt{optim(..., method = "L-BFGS-B")} function in \texttt{R}.  
The same approach is also applicable to the distance-based test statistics due to the equivalence between the CED and CMMD established in Theorem \ref{thm1}.
For the bandwidth parameters $b_l$ used in the local bootstrap, following the aforementioned references, we adopt the rule of thumb $b_l(s) = \min \{ \widehat{\sigma}_{X^{(l)}(s)}, \mathrm{IQR}_{X^{(l)}(s)} / 1.34 \} \cdot n_l^{-1/(p + 2\nu)} $ for $s = 1, \dots, p$, where $\widehat{\sigma}_{X^{(l)}(s)}$ and $\mathrm{IQR}_{X^{(l)}(s)}$ are the sample standard deviation and interquartile range of $X^{(l)}(s)$, respectively.

Following Remark \ref{remark:time}, the time complexity of both our global and local tests is $O(B(n_{1}+n_{2})^{2})$. This is the same as that of unconditional two-sample tests based on permutation approach when the number of random permutations is $B$ \citep{szekely2004testing,gretton2012kernel}.

\section{Numerical studies}\label{sec:num}
We conduct numerical studies to evaluate the performance of our proposed tests in this section.
The code that implements our tests and reproduces all numerical results is available at \url{https://github.com/lizhuoxi-97/TCDT}.

\subsection{Monte Carlo simulations}\label{subsec:simu}
In this subsection, we assess the finite-sample performance of our proposed methods based on the conditional energy distance (denoted as \texttt{CED}) and conditional maximum mean discrepancy (denoted as \texttt{CMMD}) through simulation examples.
In Examples \ref{eg:local__univariate}-\ref{eg:local__multivariate}, we examine \texttt{CED} and \texttt{CMMD} in testing the local problem \eqref{eq:hypo2}, for which no competing methods currently exist. In Examples \ref{eg:global__covariate_shift__univariate}-\ref{eg:global__label_shift}, we evaluate them in testing the global problem \eqref{eq:hypo}, comparing them against the conformal prediction test of \citet{hu2024two} (denoted as \texttt{CONF}).

In our proposed tests, we specify the semimetric $\rho$ for \texttt{CED} as the Euclidean distance. For the reproducing kernel $k$ in \texttt{CMMD}, we use the Gaussian kernel $k(y, y') = \exp\{-\|y - y'\|^2/(2\gamma^2)\}$, with the bandwidth $\gamma$ determined by the median heuristic \citep{gretton2012kernel,ke2020expected}.
For the smoothing kernel $g$, we adopt the Gaussian kernel.
In Examples \ref{eg:local__univariate}-\ref{eg:global__covariate_shift__univariate} and \ref{eg:global__label_shift}, we use a standard second-order ($\nu = 2$) kernel $g(u) = (2\pi)^{-1/2} \exp(-u^2 / 2)$.
In Example \ref{eg:global__covariate_shift__multivariate}, we employ a fourth-order ($\nu = 4$) kernel $g(u) = (3/2 - u^2/2) (2\pi)^{-1/2} \exp(-u^2 / 2)$ to handle the higher dimension $p = 4$, as discussed in Section \ref{sec:global}. 
We calculate the $p$-value of the proposed methods via the local bootstrap procedure with $B=299$ replications and the Gaussian kernel $g(u) = (2\pi)^{-1/2} \exp(-u^2 / 2)$ in (\ref{eq:P_local_boot}).
The selection of bandwidth parameters, $h_l$ and $b_l$, is provided in Section \ref{sec:applications}.

For \texttt{CONF}, we consider an equal data-splitting ratio and estimate the marginal and conditional density ratios by kernel logistic regression.
The tuning parameter $\sigma^2$ therein is set to $200$, and $\lambda$ is selected using the out-of-sample cross entropy loss, as \citet{hu2024two} suggested.

\begin{example}\label{eg:local__univariate}

To examine the performance of the proposed local tests for problem \eqref{eq:hypo2}, we focus on the general regression models
\begin{align}
Y^{(l)} = m^{(l)}(X^{(l)}) + \epsilon^{(l)} \quad (l = 1,2), \label{eq:reg_framework}
\end{align}
where $m^{(l)}$ is the regression function and $\epsilon^{(l)}$ is the error satisfying $\E(\epsilon^{(l)}\mid X^{(l)})=0$. Let $X^{(1)} \sim \N(0, 1)$ and $X^{(2)} \sim t(5)$, that is, $P_X^{(1)} \ne P_X^{(2)}$.
The error $\epsilon^{(l)} \sim \N(0, 0.5)$ is independent of $X^{(l)}$ for both $l=1,2$.
For the regression functions, we consider $m^{(1)}(x) = 0.5 x$ and $m^{(2)}(x) = - 0.5 x$. 
It is under the null hypothesis only at $x = 0$, and the signal for the difference between the two samples gets larger when the magnitude of local point $x$ is larger.

We set $n_1 = n_2 = 150$ and conducted the proposed tests at local points ranging from $-1$ to $1$, spaced by $0.1$.
For each point, we obtained empirical power or size from 1000 simulations at a significance level of $\alpha=0.05$.
Figure \ref{fig:power__local__univariate} summarizes the results.
At $x = 0$, the empirical sizes for \texttt{CED} and \texttt{CMMD} are $0.046$ and $0.056$, respectively, which are controlled around the nominal level.
The power of both \texttt{CED} and \texttt{CMMD} increases as the local signal strengthens.

\begin{figure}[htbp]
    \centering
    \includegraphics[width=0.88\textwidth]{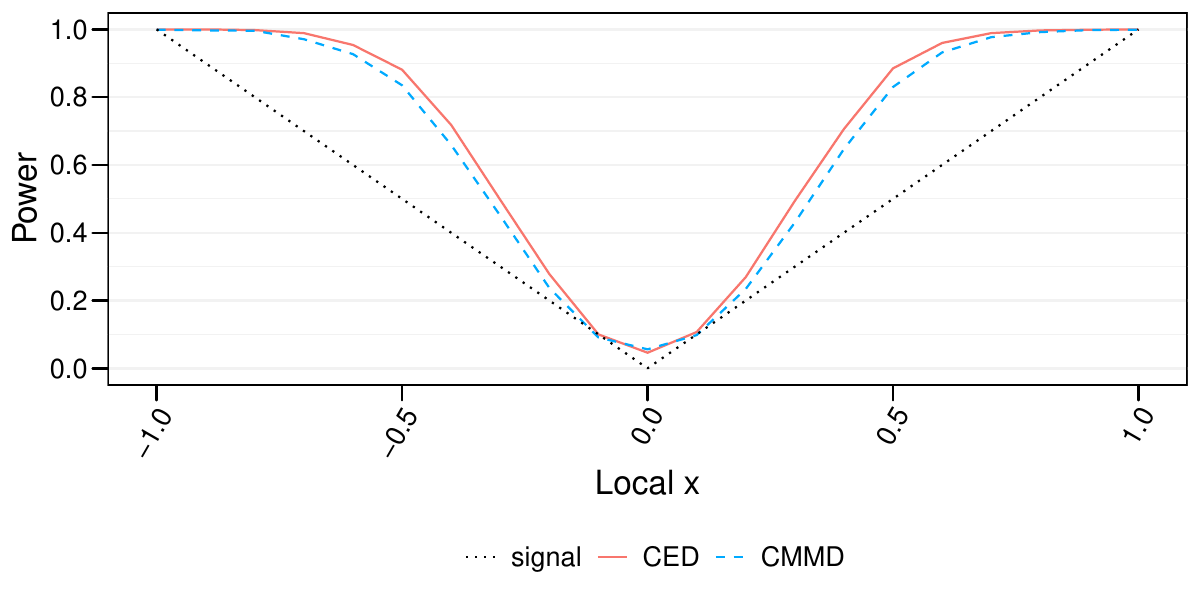}
    \caption{
    Empirical power curves under $\alpha = 0.05$, and properly scaled signals in Example \ref{eg:local__univariate}.
    }
    \label{fig:power__local__univariate}
\end{figure}

\end{example}

\begin{example}\label{eg:local__multivariate}

We further illustrate how the local tests perform in a single data generation.
Let $X^{(1)}$ follow a bivariate truncated standard normal distribution with support on $[-2, 2] \times [-2, 2]$, and let $X^{(2)} \sim \U(-2, 2) \times \U(-2, 2)$.
Define $Y^{(1)} = \|X^{(1)}\|^2 + \epsilon^{(1)}$ and $Y^{(2)} = \|X^{(2)}\|^2 \mathbbm{1}\{ X^{(2)}(1) X^{(2)}(2) \ge 0 \} - \|X^{(2)}\|^2 \mathbbm{1}\{ X^{(2)}(1) X^{(2)}(2) < 0 \} + \epsilon^{(2)}$, where $\epsilon^{(l)} \sim \N(0, 1)$ is independent of $X^{(l)}$ for $l = 1, 2$. 
In this setting, the null hypothesis in \eqref{eq:hypo2} holds whenever the local point $x$ lies in the first or third quadrant, and is violated when $x$ lies in the second or fourth quadrant.

We set $n_1 = n_2 = 500$ and conducted the local tests on a fixed uniform grid of $21 \times 21$ points over $(x(1), x(2)) \in [-2, 2] \times [-2, 2]$ at a significance level $\alpha=0.05$.
The results, based on a single data generation, are shown in Figure \ref{fig:local__grid}.
It demonstrates that the proposed tests correctly identify the regions of the alternative hypothesis where the signal is strong enough. 
The few false rejections in the first and third quadrants are due to the randomness of generating data a single time.

\begin{figure}[htbp]
    \centering
    \includegraphics[width=0.88\textwidth]{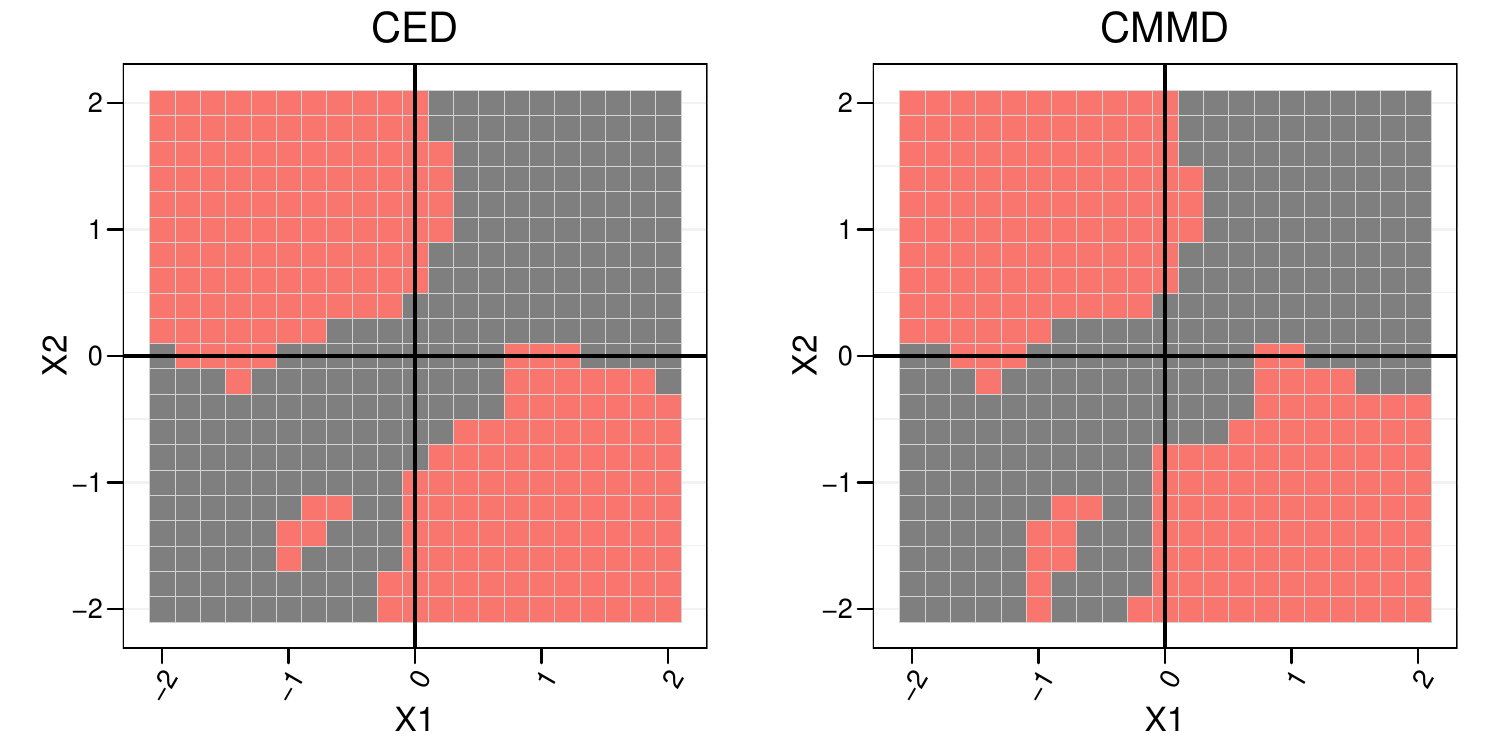}
    \caption{
    Significant local regions for Example \ref{eg:local__multivariate}, obtained by the proposed tests (\texttt{CED} and \texttt{CMMD}) under $\alpha = 0.05$ based on one-time data generation.
    The rejection regions are filled in red, whereas the non-significant regions are filled in gray.
    }
    \label{fig:local__grid}
\end{figure}
    
\end{example}

\begin{example} \label{eg:global__covariate_shift__univariate}
This example assesses the performance of the proposed global tests for problem \eqref{eq:hypo}. 
We adopt the general regression models \eqref{eq:reg_framework} from Example \ref{eg:local__univariate}. 
Let $X^{(1)} \sim \N(0, 1)$ and $X^{(2)} \sim t(5)$, resulting in $P_X^{(1)} \ne P_X^{(2)}$.
In addition to \texttt{CONF}, we compare our methods with the regression curve testing via empirical characteristic function \citep{pardo2015nonparametric} (denoted as \texttt{ECF}), which is applicable in the univariate case.

We first consider the scenarios in which the two conditional distributions differ solely in the conditional means (i.e., $m^{(l)}$). For both $l=1,2$, let $\epsilon^{(l)} \sim \N(0, 0.5)$ be independent of $X^{(l)}$. Two settings of $m^{(l)}$ are considered as follows:
\begin{setting}\label{enum:global__covariate_shift__univariate__diffm_x2+x(x+2)(x-2)} 
$m^{(1)}(x) = 1 + x^2$, $m^{(2)}(x) = 1 + x^2 + \frac{1}{5}c x(x+2)(x-2)$.
\end{setting}
\begin{setting}\label{enum:global__covariate_shift__univariate__diffm_exp+sin} 
$m^{(1)}(x) = 1 + \exp(x)$, $m^{(2)}(x) = 1 + \exp(x) + c \sin(2\pi x)$.
\end{setting}
Here, $c$ is a parameter that controls the strength of the signal.
The case $c = 0$ corresponds to the null hypothesis, while a larger $c$ indicates a greater difference.

Besides, we consider the scenarios in which the two conditional distributions differ only in the conditional variances.
Let the regression functions $m^{(1)}(x) = m^{(2)}(x) = 1 + x^2$ be identical across the two samples.
The following two cases are considered:
\begin{setting}\label{enum:global__covariate_shift__univariate__diffv_homo}
$\epsilon^{(1)} \sim \N(0, 0.25)$, $\epsilon^{(2)} \sim \N(0, 0.25 + c)$.
\end{setting}
\begin{setting}\label{enum:global__covariate_shift__univariate__diffv_hetero1} 
$\epsilon^{(1)} \sim \N(0, 0.25^2 \exp(2X^{(1)}))$, $\epsilon^{(2)} \sim \N(0, (0.25 + 0.5c)^2 \exp(2X^{(2)}))$.
\end{setting}

We set $n_1 = n_2 = 50$ or $100$ for Settings \ref{enum:global__covariate_shift__univariate__diffm_x2+x(x+2)(x-2)}-\ref{enum:global__covariate_shift__univariate__diffm_exp+sin}, and set $n_1 = n_2 = 100$ or $200$ for Settings \ref{enum:global__covariate_shift__univariate__diffv_homo}-\ref{enum:global__covariate_shift__univariate__diffv_hetero1}.
The signal strength $c$ is varied from $0$ to $2$ in Settings \ref{enum:global__covariate_shift__univariate__diffm_x2+x(x+2)(x-2)}-\ref{enum:global__covariate_shift__univariate__diffm_exp+sin}, and varied from $0$ to $1$ in Settings \ref{enum:global__covariate_shift__univariate__diffv_homo}-\ref{enum:global__covariate_shift__univariate__diffv_hetero1}.
Table \ref{tab:global__covariate_shift__univariate} and Figure \ref{fig:global__covariate_shift__univariate} summarize the empirical sizes and powers at a significance level $\alpha = 0.05$ via 1000 simulations.

\begin{table}[htbp]
    \centering
    \caption{Empirical size ($c = 0$) under $\alpha = 0.05$ in Example \ref{eg:global__covariate_shift__univariate}.
    }
    \label{tab:global__covariate_shift__univariate}
    \begin{tabular}{clrrrrr}
    \toprule
    Setting & \multicolumn{1}{c}{Sample Size} &       & \multicolumn{1}{c}{CED} & \multicolumn{1}{c}{CMMD} & \multicolumn{1}{c}{CONF} & \multicolumn{1}{c}{ECF} \\
    \midrule
    \multicolumn{7}{c}{Difference in conditional means} \\
    \midrule
    \multirow{2}[2]{*}{\ref{enum:global__covariate_shift__univariate__diffm_x2+x(x+2)(x-2)}} & $n_1 = n_2 = 50$ &       & 0.046  & 0.053  & 0.074  & 0.113  \\
        & $n_1 = n_2 = 100$ &       & 0.049  & 0.049  & 0.066  & 0.081  \\
    \midrule
    \multirow{2}[2]{*}{\ref{enum:global__covariate_shift__univariate__diffm_exp+sin}} & $n_1 = n_2 = 50$ &       & 0.043  & 0.046  & 0.043  & 0.101  \\
        & $n_1 = n_2 = 100$ &       & 0.043  & 0.046  & 0.037  & 0.084  \\
    \midrule
    \multicolumn{7}{c}{Difference in conditional variances} \\
    \midrule
    \multirow{2}[2]{*}{\ref{enum:global__covariate_shift__univariate__diffv_homo}} & $n_1 = n_2 = 100$ &       & 0.042  & 0.043  & 0.071  & 0.094  \\
        & $n_1 = n_2 = 200$ &       & 0.040  & 0.037  & 0.093  & 0.080  \\
    \midrule
    \multirow{2}[2]{*}{\ref{enum:global__covariate_shift__univariate__diffv_hetero1}} & $n_1 = n_2 = 100$ &       & 0.027  & 0.027  & 0.066  & 0.042  \\
        & $n_1 = n_2 = 200$ &       & 0.023  & 0.017  & 0.082  & 0.030  \\
    \bottomrule
    \end{tabular}%
\end{table}

\begin{figure}[htbp]
    \centering
    \includegraphics[width=\textwidth]{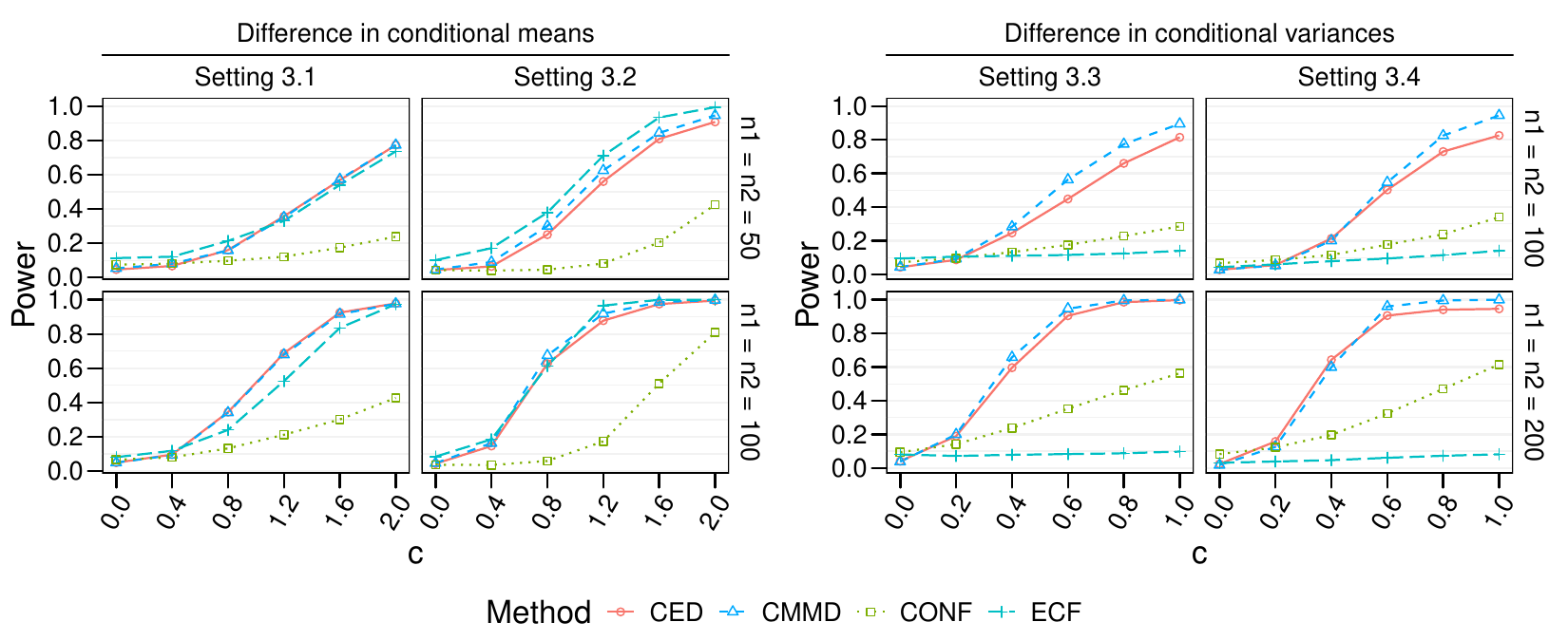}
    \caption{Empirical size ($c = 0$) and power ($c\neq 0$) under $\alpha = 0.05$ in Example \ref{eg:global__covariate_shift__univariate}.
    }
    \label{fig:global__covariate_shift__univariate}
\end{figure}

It can be observed that \texttt{CONF} and \texttt{ECF} generally yield empirical sizes higher than 0.05, whereas both \texttt{CED} and \texttt{CMMD} control the empirical size well across all scenarios.
A possible reason for \texttt{CONF}'s inflated size is its reliance on high-quality density ratio estimators and a potential requirement for unbalanced sample splitting; see Assumption 2(b) and discussion in \citet{hu2024two}.
Although \texttt{ECF} is applicable to general distributions, its size control is often unsatisfactory when the support of $X^{(l)}$ is not the interval $[0, 1]$. In Settings \ref{enum:global__covariate_shift__univariate__diffm_x2+x(x+2)(x-2)}-\ref{enum:global__covariate_shift__univariate__diffm_exp+sin} where differences in complex nonlinear regression functions are considered, the power of \texttt{CED}, \texttt{CMMD} and \texttt{ECF} are notably greater than that of \texttt{CONF}.
In Settings \ref{enum:global__covariate_shift__univariate__diffv_homo}-\ref{enum:global__covariate_shift__univariate__diffv_hetero1}, the proposed tests also show much higher power than \texttt{CONF}, while \texttt{ECF} shows trivial power as it is designed for comparing regression curves.
The satisfactory performance of \texttt{CED} and \texttt{CMMD} in all scenarios highlights their applicability to various types of alternative hypotheses.
\end{example}

\begin{example}\label{eg:global__covariate_shift__multivariate}
Similar to Example \ref{eg:global__covariate_shift__univariate}, we now consider the setting where the dimension of $X$ is $p = 4$.
The data generating process again follows the regression models \eqref{eq:reg_framework}. 
Let $X^{(1)} \sim \N(\mathbf{0}, \mathbf{I})$ and $X^{(2)} \sim \N(\boldsymbol{\mu}, \mathbf{I})$, where $\boldsymbol{\mu} = (1, 1, -1, 0)^{\top}$.
For both $l=1,2$, the error $\epsilon^{(l)} \sim t(5)$ has a heavy tail and is independent of $X^{(l)}$.
We set $m^{(1)}(x) = \{\theta(x)\}^2$ and $m^{(2)}(x) = \{\theta(x)\}^2 + c \theta(x)$, where $\theta(x) = x(1) + x(2) + x(3) + x(4) + 1$ and $c$ is the signal strength.

We let $n_1 = n_2 = 100$ or $200$ and vary $c$ from $0$ to $2$.
The simulations are conducted 1000 times at a significance level of $\alpha=0.05$.
Table \ref{tab:size__global__covariate_shift__multivariate} and Figure \ref{fig:power__global__covariate_shift__multivariate} present the empirical sizes and powers.
In this multivariate $X$ setting, both \texttt{CED} and \texttt{CMMD} are observed to achieve better power performance compared to \texttt{CONF}, while guaranteeing type I error control.

\begin{table}[htbp]
\centering
\caption{Empirical size ($c = 0$) under $\alpha = 0.05$ in Example \ref{eg:global__covariate_shift__multivariate}.}
\label{tab:size__global__covariate_shift__multivariate}
    \begin{tabular}{lrrrr}
    \toprule
    \multicolumn{1}{c}{Sample Size} &       & \multicolumn{1}{c}{CED} & \multicolumn{1}{c}{CMMD} & \multicolumn{1}{c}{CONF} \\
    \midrule
    $n_1 = n_2 = 100$ &       & 0.016  & 0.025  & 0.037  \\
    $n_1 = n_2 = 200$ &       & 0.052  & 0.049  & 0.040  \\
    \bottomrule
    \end{tabular}%
\end{table}

\begin{figure}[htbp]
    \centering
    \includegraphics[width=0.88\textwidth]{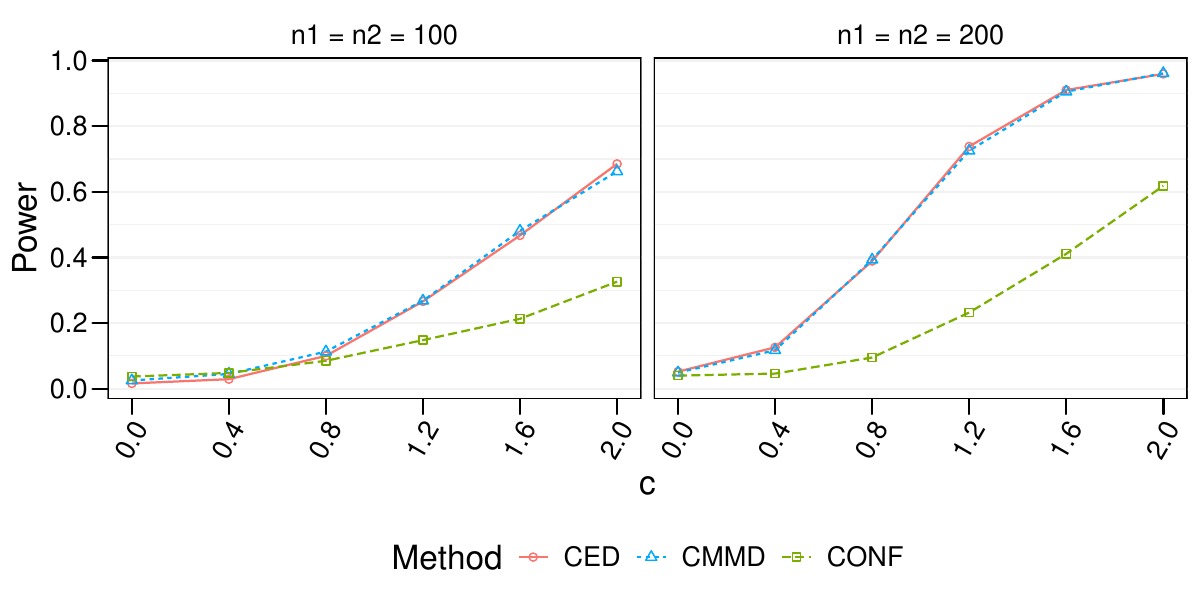}
    \caption{Empirical size ($c = 0$) and power ($c\neq 0$) under $\alpha = 0.05$ in Example \ref{eg:global__covariate_shift__multivariate}.}
    \label{fig:power__global__covariate_shift__multivariate}
\end{figure}

\end{example}

\begin{example}\label{eg:global__label_shift}
We further consider a scenario that appears in transfer learning: testing whether the conditional distribution of covariates given the response, $P_{X\mid Y}$, is the same for two populations (i.e., prior shift assumption).
To generate data, we adopt a multivariate regression model of $X$ on $Y$, conversely to \eqref{eq:reg_framework}:
\begin{align*}
    X^{(l)} = m^{(l)}(Y^{(l)}) + \epsilon^{(l)} \quad (l = 1, 2),
\end{align*}
where $m^{(l)}: \R \to \R^{p}$ is the multivariate regression function and $\epsilon^{(l)} \in \R^{p}$ is the multivariate error satisfying $\E(\epsilon^{(l)}\mid Y^{(l)})=\mathbf{0}$.
Let $\epsilon^{(l)} \sim \N(\mathbf{0}, \pmb{\Sigma})$ be independent of $Y^{(l)}$ for both $l = 1, 2$, where $\pmb{\Sigma} = (\rho^{|i - j|})$ with $\rho = 0.5$.
We set $Y^{(1)} \sim \N(0, 1)$ and $Y^{(2)} \sim \N(0, 1.5)$, that is, $P_Y^{(1)} \ne P_Y^{(2)}$.
Let $\theta^{(1)}(y) = y \cdot \beta$ and $\theta^{(2)}(y) = (1 + c) y \cdot \beta$, where $\beta \in \R^{p}$, and $c$ controls the signal for the difference of the conditional distributions between the two samples.
We then set
\begin{equation*}
    m^{(l)}(y) = \left [\{\theta^{(l)}(y)\}^{\circ 2} + 3 \theta^{(l)}(y) + 2 \right ] \div \left [\{\theta^{(l)}(y)\}^{\circ 2} + 1 \right ]\quad (l=1,2),
\end{equation*}
where $^{\circ 2}$ denotes the Hadamard square and $\div$ denotes the Hadamard division.
We consider $p = 5$ or $20$.
For $p = 5$, we set $\beta(j) = -0.5$ for $j = 1, 2$, $\beta(j) = 0.5$ for $j = 3, 4$ and $\beta(j) = 1$ for $j = 5$.
For $p = 20$, we set $\beta(j) = -0.5$ for $j = 1, 2, 3$, $\beta(j) = 0.5$ for $j = 4, 5$, $\beta(j) = 1$ for $j = 6, \dots, 10$ and $\beta(j) = 0$ for $j = 11, \dots, 20$.

Let $n_1 = n_2 = 100$ or $200$, and vary the signal strength $c$ from $0$ to $1$. 
For each setting, we repeat 1000 simulations to obtain the empirical size ($c = 0$) and power ($c \ne 0$) of \texttt{CMMD}, \texttt{CED} and \texttt{CONF} at a significance level of $\alpha=0.05$.
The results are presented in Table \ref{tab:size__global__label_shift} and Figure \ref{fig:power__global__label_shift}.
Table \ref{tab:size__global__label_shift} demonstrates that all three methods control the type I error rate.
The power comparison, illustrated in Figure \ref{fig:power__global__label_shift}, clearly shows an advantage of \texttt{CED} and \texttt{CMMD} over \texttt{CONF}.

\begin{table}[htbp]
    \centering    
    \caption{Empirical size ($c = 0$) under $\alpha = 0.05$ in Example \ref{eg:global__label_shift}.
    }\label{tab:size__global__label_shift}
    \begin{tabular}{lrrrrrrrr}
    \toprule
          &       & \multicolumn{3}{c}{$p = 5$} &       & \multicolumn{3}{c}{$p = 20$} \\
    \cmidrule{3-5}\cmidrule{7-9}\multicolumn{1}{c}{Sample Size} &       & \multicolumn{1}{c}{CED} & \multicolumn{1}{c}{CMMD} & \multicolumn{1}{c}{CONF} &       & \multicolumn{1}{c}{CED} & \multicolumn{1}{c}{CMMD} & \multicolumn{1}{c}{CONF} \\
    \midrule
    $n_1 = n_2 = 100$ &       & 0.026  & 0.026  & 0.051  &       & 0.037  & 0.037  & 0.039  \\
    $n_1 = n_2 = 200$ &       & 0.050  & 0.054  & 0.049  &       & 0.033  & 0.033  & 0.019  \\
    \bottomrule
    \end{tabular}%
\end{table}

\begin{figure}[htbp]
    \centering
    \includegraphics[width=0.88\textwidth]{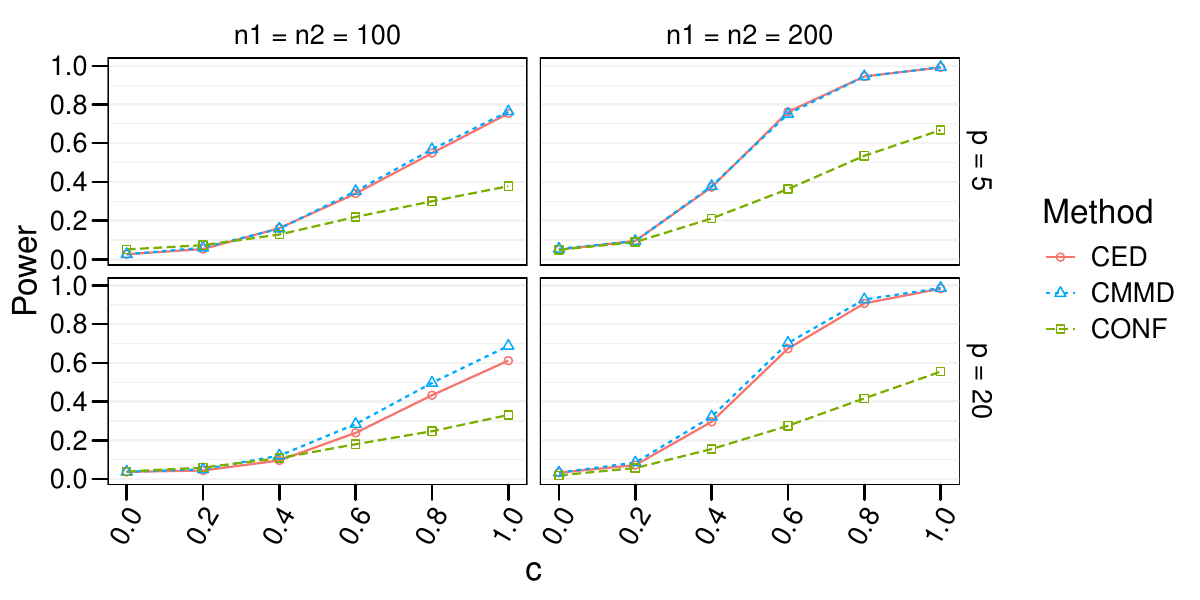}
    \caption{Empirical size ($c = 0$) and power ($c\neq 0$) under $\alpha = 0.05$ in Example \ref{eg:global__label_shift}.
    }
    \label{fig:power__global__label_shift}
\end{figure}
    
\end{example}

\subsection{Airfoil data example}

We consider the airfoil dataset from the UCI Machine Learning Repository \citep{airfoil}. 
This dataset, collected by NASA to study airfoil sound pressure, has been previously analyzed in \citet{tibshirani2019conformal,hu2024two,huang2024efficient}. 
It contains $N=1503$ observations, with a response $Y$ (scaled sound pressure level) and covariates $X$ with $p=5$: log-frequency, angle of attack, chord length, free-stream velocity, and suction-side log-displacement thickness.

As the dataset consists of a single sample only, we adopt a semi-synthetic approach, similar to that of \citet{tibshirani2019conformal,hu2024two,huang2024efficient}, to create two-sample scenarios. 
We generate four different settings by partitioning the data:
\begin{enumerate}
    \item\label{enum:realdata2_cov_shift_null} Covariate shift under the null hypothesis:
    The dataset is randomly partitioned into $\mathcal{D}_1$ with $n_1 = 301$ and $\mathcal{\widetilde{D}}_2$ with $\tilde{n}_2 = 1202$. 
    The second sample, $\mathcal{D}_2$, is formed by sampling $n_2=301$ points from $\mathcal{\widetilde{D}}_2$ without replacement, with probabilities proportional to $w(X) = \exp(\alpha^\top X)$, where $\alpha = (-1, 0, 0, 0, 1)^\top$. 
    This induces a covariate shift while preserving $P_{Y\mid X}$.
    \item\label{enum:realdata2_prior_shift_null} Prior shift under the null hypothesis:
    The procedure is identical to Setting \ref{enum:realdata2_cov_shift_null}, but sampling probabilities are instead proportional to the response, $w(Y) = Y$. 
    This induces a prior shift while preserving $P_{X\mid Y}$.
    \item\label{enum:realdata2_cov_shift_alt} Covariate shift under the alternative hypothesis:
    The dataset is partitioned into two groups based on the median of the response variable.
    The first sample contains $n_1 = 752$ points with smaller values of $Y$, while the second sample contains the rest of $n_2 = 751$ points. 
    To avoid singularity between two groups, we randomly select $0.05n_1$ points in each of the two samples and then flip their groups.
    This construction is designed to violate the equality of $P_{Y\mid X}$.
    \item\label{enum:realdata2_prior_shift_alt} Prior shift under the alternative hypothesis:
    The setup is almost the same as Setting \ref{enum:realdata2_cov_shift_alt}, but the partition is based on the median of the ``chord length" variable, creating a scenario where $P_{X\mid Y}$ differs.
\end{enumerate}

For each setting, we apply our global tests at a significance level of $\alpha = 0.05$.
In the covariate shift scenarios (Settings \ref{enum:realdata2_cov_shift_null} and \ref{enum:realdata2_cov_shift_alt}), we test the equality of $P_{Y\mid X}$ using a Gaussian smoothing kernel of $\nu = 4$, as here $p=5$. 
In the prior shift scenarios (Settings \ref{enum:realdata2_prior_shift_null} and \ref{enum:realdata2_prior_shift_alt}), we test the equality of $P_{X\mid Y}$ with a Gaussian smoothing kernel of $\nu = 2$.
All other implementation details are the same as those in Section \ref{subsec:simu}.

The results are summarized in Table \ref{tab:realdata2}. 
Under the null hypothesis (Settings \ref{enum:realdata2_cov_shift_null}-\ref{enum:realdata2_prior_shift_null}), we performed 500 replications and report the empirical rejection rates. 
The type I error of both \texttt{CED} and \texttt{CMMD} are close to the nominal 0.05 level. 
Under the alternative hypothesis (Settings \ref{enum:realdata2_cov_shift_alt}-\ref{enum:realdata2_prior_shift_alt}), which involve deterministic partitions, we report the $p$-values from a single test. 
In both cases, the tests correctly reject $H_0$ with very small $p$-values.

\begin{table}[htbp]
    \centering
    \caption{Airfoil data example: Empirical rejection rate and $p$-values from the proposed global tests for different partition settings.}\label{tab:realdata2}
    \begin{tabular}{lrrrlrr}
    \toprule
        & \multicolumn{2}{c}{Rejection rate} &       &       & \multicolumn{2}{c}{$p$-value} \\
    \cmidrule{2-3}\cmidrule{6-7}      & \multicolumn{1}{c}{CED} & \multicolumn{1}{c}{CMMD} &       &       & \multicolumn{1}{c}{CED} & \multicolumn{1}{c}{CMMD} \\
    \cmidrule{1-3}\cmidrule{5-7}Setting \ref{enum:realdata2_cov_shift_null} & 0.052  & 0.052  &       & Setting \ref{enum:realdata2_cov_shift_alt} & 0.003  & 0.003  \\
    Setting \ref{enum:realdata2_prior_shift_null} & 0.044  & 0.042  &       & Setting \ref{enum:realdata2_prior_shift_alt} & 0.003  & 0.003  \\
    \bottomrule
    \end{tabular}%
\end{table}

Furthermore, we apply our local tests to Setting \ref{enum:realdata2_prior_shift_alt}. 
Specifically, we test for local differences in $P_{X\mid Y=y}$ for values of the response $y$ ranging from 110 to 130. 
Figure \ref{fig:airfoil} plots the resulting $p$-values against $y$. 
The $p$-values remain consistently below the 0.05 significance level across the entire range, in agreement with the setup of Setting \ref{enum:realdata2_prior_shift_alt}.

\begin{figure}[htbp]
    \centering
    \includegraphics[width=0.88\textwidth]{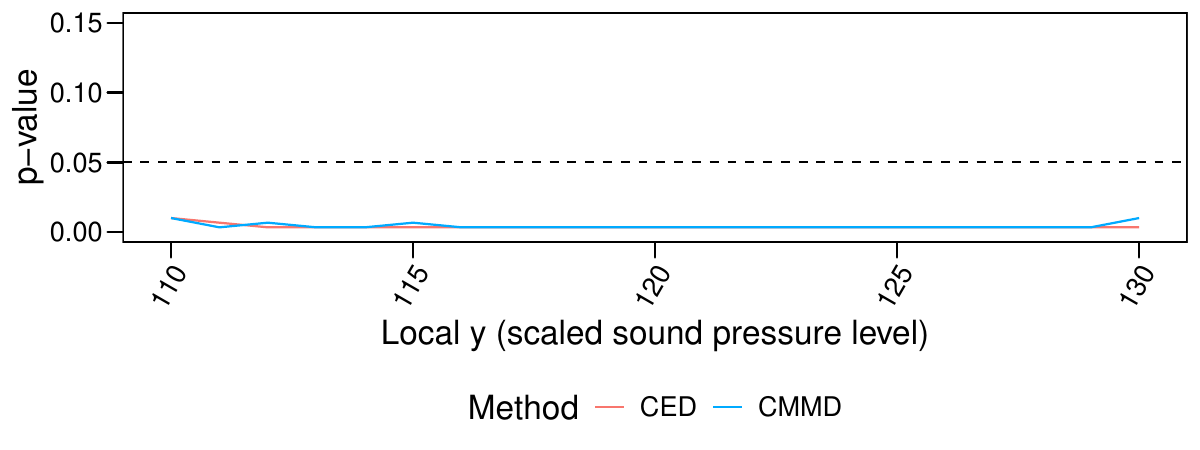}
    \caption{Airfoil data example: Curve of $p$-values from the proposed local tests on Setting \ref{enum:realdata2_prior_shift_alt} for local $y$ between 110 and 130.
    The horizontal dashed line indicates the $0.05$ significance level.
    }
    \label{fig:airfoil}
\end{figure}

\section{Discussions}
Several research directions warrant further exploration. First, it would be valuable to extend the framework to the $K$-sample conditional distribution testing problem for $K\geq 2$ and establish connections with distance-based multivariate analysis of variance \citep{rizzo2010disco}. 
Beyond kernel smoothing, alternative machine learning techniques, such as random forest and neural networks, could be considered for estimating the CED and ICED. These methods are anticipated to yield improved performance for higher-dimensional $X$, but deriving the asymptotic distributions of the resulting test statistics presents significant theoretical challenges. 
Finally, for local testing problems, especially when inference at a boundary point is of interest, a statistic based on local polynomial fitting may be preferable to the sample CED. Such an approach could be particularly useful in the regression discontinuity design \citep{calonico2014robust}.


\newpage
\begin{appendices}
\counterwithin{theorem}{section}
\counterwithin{definition}{section}
\renewcommand{\theequation}{A.\arabic{equation}}
\setcounter{equation}{0}

\section{Computation of the sample CED and ICED}\label{app:a}
For the sample CED $\widehat{\D}_{\rho}(x)$, note that, for $i=1,\ldots,n_{1}$,
\[
\sum_{i'\neq i}G_{h_{1}}(X_{i'}^{(1)}-x)=\sum_{i'=1}^{n_{1}}G_{h_{1}}(X_{i'}^{(1)}-x)-G_{h_{1}}(X_{i}^{(1)}-x),
\]
and, for $j=1,\ldots,n_{2}$,
\[
\sum_{j'\neq j}G_{h_{2}}(X_{j'}^{(2)}-x)=\sum_{j'=1}^{n_{2}}G_{h_{2}}(X_{j'}^{(2)}-x)-G_{h_{2}}(X_{j}^{(2)}-x).
\]
We first calculate $\sum_{i'=1}^{n_{1}}G_{h_{1}}(X_{i'}^{(1)}-x)$ and $\sum_{j'=1}^{n_{2}}G_{h_{2}}(X_{j'}^{(2)}-x)$, which can be computed in $O(n_{1})$ and $O(n_{2})$, respectively. Then it can be seen that the overall cost of computing the statistic $\widehat{\D}_{\rho}(x)$ is $O((n_{1}+n_{2})^{2})$. 

For the sample ICED $\widehat{\I}_{\rho}$, note that, for $i_{1}=1,\ldots,n_{1}$,
\[
\sum_{j_{1}\neq j_{2}}G_{h_{2}}(X_{j_{1}}^{(2)}-X_{i_{1}}^{(1)})G_{h_{2}}(X_{j_{2}}^{(2)}-X_{i_{1}}^{(1)})=\bigg[\sum_{j=1}^{n_{2}}G_{h_{2}}(X_{j}^{(2)}-X_{i_{1}}^{(1)})\bigg]^{2}-\sum_{j=1}^{n_{2}}\Big[G_{h_{2}}(X_{j}^{(2)}-X_{i_{1}}^{(1)})\Big]^{2},
\]
for $j_{1}=1,\ldots,n_{2}$,
\[
\sum_{i_{1}\neq i_{2}}G_{h_{1}}(X_{i_{1}}^{(1)}-X_{j_{1}}^{(2)})G_{h_{1}}(X_{i_{2}}^{(1)}-X_{j_{1}}^{(2)})=\bigg[\sum_{i=1}^{n_{1}}G_{h_{1}}(X_{i}^{(1)}-X_{j_{1}}^{(2)})\bigg]^{2}-\sum_{i=1}^{n_{1}}\Big[G_{h_{1}}(X_{i}^{(1)}-X_{j_{1}}^{(2)})\Big]^{2},
\]
and, for $i=1,\ldots,n_{1}$ and $j=1,\ldots,n_{2}$, 
\begin{align*}
    \sum_{i'\neq i}G_{h_{1}}(X_{i'}^{(1)}-X_{j}^{(2)})&=\sum_{i'=1}^{n_{1}}G_{h_{1}}(X_{i'}^{(1)}-X_{j}^{(2)})-G_{h_{1}}(X_{i}^{(1)}-X_{j}^{(2)}),\\
    \sum_{j'\neq j}G_{h_{2}}(X_{j'}^{(2)}-X_{i}^{(1)})&=\sum_{j'=1}^{n_{2}}G_{h_{2}}(X_{j'}^{(2)}-X_{i}^{(1)})-G_{h_{2}}(X_{j}^{(2)}-X_{i}^{(1)}).
\end{align*}
We first calculate
\begin{align*}
    \bigg\{\sum_{j=1}^{n_{2}}G_{h_{2}}(X_{j}^{(2)}-X_{i}^{(1)})\bigg\}_{i=1}^{n_{1}},\quad\bigg\{\sum_{j=1}^{n_{2}}\Big[G_{h_{2}}(X_{j}^{(2)}-X_{i}^{(1)})\Big]^{2}\bigg\}_{i=1}^{n_{1}},\\
    \bigg\{\sum_{i=1}^{n_{1}}G_{h_{1}}(X_{i}^{(1)}-X_{j}^{(2)})\bigg\}_{j=1}^{n_{2}},\quad\bigg\{\sum_{i=1}^{n_{1}}\Big[G_{h_{1}}(X_{i}^{(1)}-X_{j}^{(2)})\Big]^{2}\bigg\}_{j=1}^{n_{2}}.
\end{align*}
Each of these sequences can be computed with $O(n_{1}n_{2})$ operations. Then, it can be seen that the overall cost of computing the statistic $\widehat{\I}_{\rho}$ is $O((n_{1}+n_{2})^{2})$.

\section{Ancillary results}
Below we list some relevant results about the generalized U-statistic in \citet{lee2019u}. 
\begin{definition}[Section 2.2, \citet{lee2019u}]\label{defa1}
    Assume that $\{Z_{i}^{(1)}\}_{i=1}^{n_{1}}$ and $\{Z_{j}^{(2)}\}_{j=1}^{n_{2}}$ are \iid\ samples from distributions $P^{(1)}$ and $P^{(2)}$ respectively, and $\{Z_{i}^{(1)}\}_{i=1}^{n_{1}}$ is independent of $\{Z_{j}^{(2)}\}_{j=1}^{n_{2}}$. Let $\psi$ be a function of $m_{1}+m_{2}$ arguments
    \[
    \psi(z_{1}^{(1)},\ldots,z_{m_{1}}^{(1)};z_{1}^{(2)},\ldots,z_{m_{2}}^{(2)}),
    \]
    which is symmetric in $z_{1}^{(1)},\ldots,z_{m_{1}}^{(1)}$ and $z_{1}^{(2)},\ldots,z_{m_{2}}^{(2)}$. The generalized U-statistic based on $\psi$ is a statistic of the form
    \[
    U_{n_{1},n_{2}}=\binom{n_{1}}{m_{1}}^{-1}\binom{n_{2}}{m_{2}}^{-1}\sum_{(n_{1},m_{1})}\sum_{(n_{2},m_{2})}\psi(Z_{i_{1}}^{(1)},\ldots,Z_{i_{m_{1}}}^{(1)};Z_{j_{1}}^{(2)},\ldots,Z_{j_{m_{2}}}^{(2)}), 
    \]
    where the summation is over all $m_{1}$-subsets of $\{Z_{i}^{(1)}\}_{i=1}^{n_{1}}$ and $m_{2}$-subsets of $\{Z_{j}^{(2)}\}_{j=1}^{n_{2}}$. Then $U_{n_{1},n_{2}}$ is an unbiased estimator of $\E\{\psi(Z_{1}^{(1)},\ldots,Z_{m_{1}}^{(1)};Z_{1}^{(2)},\ldots,Z_{m_{2}}^{(2)})\}$. 
\end{definition}

\begin{definition}[Section 2.2, \citet{lee2019u}]\label{defa2}
    For $c=0,\ldots,m_{1}$ and $d=0,\ldots,m_{2}$, define
    \[
    \psi_{c,d}(z_{1}^{(1)},\ldots,z_{c}^{(1)};z_{1}^{(2)},\ldots,z_{d}^{(2)})=\E\{\psi(z_{1}^{(1)},\ldots,z_{c}^{(1)},Z_{c+1}^{(1)},\ldots,Z_{m_{1}}^{(1)};z_{1}^{(2)},\ldots,z_{d}^{(2)},Z_{d+1}^{(2)},\ldots,Z_{m_{2}}^{(2)})\},
    \]
    and
    \[
    \sigma_{c,d}^{2}=\var\{\psi_{c,d}(Z_{1}^{(1)},\ldots,Z_{c}^{(1)};Z_{1}^{(2)},\ldots,Z_{d}^{(2)})\}.
    \]
\end{definition}

\begin{theorem}[Theorem 2 in Section 2.2, \citet{lee2019u}]\label{thma1}
    The variance of the generalized U-statistic $U_{n_{1},n_{2}}$ in Definition \ref{defa1} is given by
    \[
    \var(U_{n_{1},n_{2}})=\binom{n_{1}}{m_{1}}^{-1}\binom{n_{2}}{m_{2}}^{-1}\sum_{c=0}^{m_{1}}\sum_{d=0}^{m_{2}}\binom{m_{1}}{c}\binom{m_{2}}{d}\binom{n_{1}-m_{1}}{m_{1}-c}\binom{n_{2}-m_{2}}{m_{2}-d}\sigma_{c,d}^{2},
    \]
    where $\sigma_{c,d}^{2}$ is given in Definition \ref{defa2}. 
\end{theorem}

\begin{definition}[Section 2.2, \citet{lee2019u}]\label{defa3}
    Let $F_{z}$ denote the distribution function of a single point mass at $z$. For $c=0,\ldots,m_{1}$ and $d=0,\ldots,m_{2}$, define
    \begin{align*}
        &\quad\phi^{(c,d)}(z_{1}^{(1)},\ldots,z_{c}^{(1)};z_{1}^{(2)},\ldots,z_{d}^{(2)})\\
        &=\int\ldots\int\psi(u_{1},\ldots,u_{m_{1}};v_{1},\ldots,v_{m_{2}})\prod_{i=1}^{c}\{dF_{z_{i}^{(1)}}(u_{i})-dP^{(1)}(u_{i})\}\prod_{i=c+1}^{m_{1}}dP^{(1)}(u_{i})\\
        &\quad\times\prod_{j=1}^{d}\{dF_{z_{j}^{(2)}}(v_{j})-dP^{(2)}(v_{j})\}\prod_{j=d+1}^{m_{2}}dP^{(2)}(v_{j}).
    \end{align*}
\end{definition}

\begin{theorem}[Hoeffding decomposition, Theorem 3 in Section 2.2, \citet{lee2019u}]\label{thma2}
    The generalized U-statistic $U_{n_{1},n_{2}}$ in Definition \ref{defa1} admits the representation
    \[
    U_{n_{1},n_{2}}=\sum_{c=0}^{m_{1}}\sum_{d=0}^{m_{2}}\binom{m_{1}}{c}\binom{m_{2}}{d}H_{n_{1},n_{2}}^{(c,d)},
    \]
    where $H_{n_{1},n_{2}}^{(c,d)}$ is the generalized U-statistic based on $\phi^{(c,d)}$ in Definition \ref{defa3} and is given by
    \[
    H_{n_{1},n_{2}}^{(c,d)}=\binom{n_{1}}{c}^{-1}\binom{n_{2}}{d}^{-1}\sum_{(n_{1},c)}\sum_{(n_{2},d)}\phi^{(c,d)}(Z_{i_{1}}^{(1)},\ldots,Z_{i_{c}}^{(1)};Z_{j_{1}}^{(2)},\ldots,Z_{j_{d}}^{(2)}). 
    \]
    Moreover, the functions $\phi^{(c,d)}$ satisfy

    (i) $\E\{\phi^{(c,d)}(Z_{1}^{(1)},\ldots,Z_{c}^{(1)};Z_{1}^{(2)},\ldots,Z_{d}^{(2)})\}=0$;

    (ii) $\cov\{\phi^{(c,d)}(S_{1};S_{2}),\phi^{(c',d')}(S_{1}';S_{2}')\}=0$ for all integers $c,d,c',d'$ and sets $S_{1},S_{2},S_{1}',S_{2}'$ unless $c=c'$, $d=d'$, $S_{1}=S_{1}'$ and $S_{2}=S_{2}'$. 

    The generalized U-statistics $H_{n_{1},n_{2}}^{(c,d)}$ are thus all uncorrelated. Their variances are given by
    \[
    \var(H_{n_{1},n_{2}}^{(c,d)})=\binom{n_{1}}{c}^{-1}\binom{n_{2}}{d}^{-1}\delta_{c,d}^{2},
    \]
    where $\delta_{c,d}^{2}=\var\{\phi^{(c,d)}(Z_{1}^{(1)},\ldots,Z_{c}^{(1)};Z_{1}^{(2)},\ldots,Z_{d}^{(2)})\}$. 
\end{theorem}

\section{Technical details}\label{app:c}
For $l=1,2$ and $i=1,\ldots,n_{l}$, denote $Z_{i}^{(l)}=(Y_{i}^{(l)},X_{i}^{(l)})$. 
\begin{proof}[Proof of Theorem \ref{thm1}]
    The facts 1 and 2 directly follow from the definitions of semimetric of strong negative type and characteristic kernel. 
    
    Suppose $k$ generates $\rho$, i.e., $\rho(y,y')=\frac{1}{2}\{k(y,y)+k(y',y')\}-k(y,y')$. Then
    \begin{align*}
        \D_{\rho}(x)&=-\int_{\Y}\int_{\Y}\rho(y,y')d(P_{Y\mid X=x}^{(1)}-P_{Y\mid X=x}^{(2)})(y)d(P_{Y\mid X=x}^{(1)}-P_{Y\mid X=x}^{(2)})(y')\\
        &=\int_{\Y}\int_{\Y}\bigg[k(y,y')-\frac{1}{2}\{k(y,y)+k(y',y')\}\bigg]d(P_{Y\mid X=x}^{(1)}-P_{Y\mid X=x}^{(2)})(y)d(P_{Y\mid X=x}^{(1)}-P_{Y\mid X=x}^{(2)})(y')\\
        &=\int_{\Y}\int_{\Y}k(y,y')d(P_{Y\mid X=x}^{(1)}-P_{Y\mid X=x}^{(2)})(y)d(P_{Y\mid X=x}^{(1)}-P_{Y\mid X=x}^{(2)})(y')\\
        &=\gamma_{k}^{2}(x),
    \end{align*}
    where we have used the fact that $\int_{\Y}d(P_{Y\mid X=x}^{(1)}-P_{Y\mid X=x}^{(2)})(y)=0$. 
\end{proof}

\begin{proof}[Proof of Theorem \ref{thm2}]
    Define the generalized U-statistic
    \[
    U_{n_{1},n_{2}}=\binom{n_{1}}{2}^{-1}\binom{n_{2}}{2}^{-1}\sum_{1\le i_{1}<i_{2}\le n_{1}}\sum_{1\le j_{1}<j_{2}\le n_{2}}\psi(Z_{i_{1}}^{(1)},Z_{i_{2}}^{(1)};Z_{j_{1}}^{(2)},Z_{j_{2}}^{(2)}),
    \]
    where
    \begin{align*}
        &\quad\psi(Z_{1}^{(1)},Z_{2}^{(1)};Z_{1}^{(2)},Z_{2}^{(2)})\\
        &=\bigg[\frac{1}{2}\Big\{\rho(Y_{1}^{(1)},Y_{1}^{(2)})+\rho(Y_{1}^{(1)},Y_{2}^{(2)})+\rho(Y_{2}^{(1)},Y_{1}^{(2)})+\rho(Y_{2}^{(1)},Y_{2}^{(2)})\Big\}-\rho(Y_{1}^{(1)},Y_{2}^{(1)})-\rho(Y_{1}^{(2)},Y_{2}^{(2)})\\
        &\quad-\D_{\rho}(x)\bigg]G_{h_{1}}(X_{1}^{(1)}-x)G_{h_{1}}(X_{2}^{(1)}-x)G_{h_{2}}(X_{1}^{(2)}-x)G_{h_{2}}(X_{2}^{(2)}-x).
    \end{align*}
    Under Assumptions \ref{assum1}-\ref{assum3}, we have
    \begin{align*}
        &\quad\widehat{\D}_{\rho}(x)-\D_{\rho}(x)\\
        &=\bigg\{\frac{1}{n_{1}(n_{1}-1)}\sum_{i_{1}\neq i_{2}}G_{h_{1}}(X_{i_{1}}^{(1)}-x)G_{h_{1}}(X_{i_{2}}^{(1)}-x)\frac{1}{n_{2}(n_{2}-1)}\sum_{j_{1}\neq j_{2}}G_{h_{2}}(X_{j_{1}}^{(2)}-x)G_{h_{2}}(X_{j_{2}}^{(2)}-x)\bigg\}^{-1}\\
        &\quad\times U_{n_{1},n_{2}}\\
        &=\big[\{f_{1}(x)f_{2}(x)\}^{2}+o_{p}(1)\big]^{-1}U_{n_{1},n_{2}}.
    \end{align*}
    It suffices to show that
    \begin{equation}\label{thm2_pf}
        \E(U_{n_{1},n_{2}})=o(1),\quad\var(U_{n_{1},n_{2}})=o(1).
    \end{equation}
    
    For the first part of (\ref{thm2_pf}), we have
    \begin{align*}
        \E(U_{n_{1},n_{2}})&=2\E\{\rho(Y_{1}^{(1)},Y_{1}^{(2)})G_{h_{1}}(X_{1}^{(1)}-x)G_{h_{1}}(X_{2}^{(1)}-x)G_{h_{2}}(X_{1}^{(2)}-x)G_{h_{2}}(X_{2}^{(2)}-x)\}\\
        &\quad-\E\{\rho(Y_{1}^{(1)},Y_{2}^{(1)})G_{h_{1}}(X_{1}^{(1)}-x)G_{h_{1}}(X_{2}^{(1)}-x)G_{h_{2}}(X_{1}^{(2)}-x)G_{h_{2}}(X_{2}^{(2)}-x)\}\\
        &\quad-\E\{\rho(Y_{1}^{(2)},Y_{2}^{(2)})G_{h_{1}}(X_{1}^{(1)}-x)G_{h_{1}}(X_{2}^{(1)}-x)G_{h_{2}}(X_{1}^{(2)}-x)G_{h_{2}}(X_{2}^{(2)}-x)\}\\
        &\quad-\D_{\rho}(x)\E\{G_{h_{1}}(X_{1}^{(1)}-x)G_{h_{1}}(X_{2}^{(1)}-x)G_{h_{2}}(X_{1}^{(2)}-x)G_{h_{2}}(X_{2}^{(2)}-x)\}.
    \end{align*}
    We first consider the term
    \begin{align*}
        &\quad\E\{\rho(Y_{1}^{(1)},Y_{1}^{(2)})G_{h_{1}}(X_{1}^{(1)}-x)G_{h_{1}}(X_{2}^{(1)}-x)G_{h_{2}}(X_{1}^{(2)}-x)G_{h_{2}}(X_{2}^{(2)}-x)\}\\
        &=\int\rho(y_{1}^{(1)},y_{1}^{(2)})G_{h_{1}}(x_{1}^{(1)}-x)G_{h_{1}}(x_{2}^{(1)}-x)G_{h_{2}}(x_{1}^{(2)}-x)G_{h_{2}}(x_{2}^{(2)}-x)\\
        &\quad\times f_{1}(y_{1}^{(1)},x_{1}^{(1)})f_{2}(y_{1}^{(2)},x_{1}^{(2)})f_{1}(x_{2}^{(1)})f_{2}(x_{2}^{(2)})dy_{1}^{(1)}dy_{1}^{(2)}dx_{1}^{(1)}dx_{2}^{(1)}dx_{1}^{(2)}dx_{2}^{(2)}.
    \end{align*}
    Let $t=h_{1}^{-1}(x_{1}^{(1)}-x)$, $u=h_{1}^{-1}(x_{2}^{(1)}-x)$, $v=h_{2}^{-1}(x_{1}^{(2)}-x)$ and $w=h_{2}^{-1}(x_{2}^{(2)}-x)$. We have
    \begin{align*}
        &\quad\E\{\rho(Y_{1}^{(1)},Y_{1}^{(2)})G_{h_{1}}(X_{1}^{(1)}-x)G_{h_{1}}(X_{2}^{(1)}-x)G_{h_{2}}(X_{1}^{(2)}-x)G_{h_{2}}(X_{2}^{(2)}-x)\}\\
        &\overset{(i)}{=}\int\rho(y_{1}^{(1)},y_{1}^{(2)})\prod_{s=1}^{p}\{g(t(s))g(u(s))g(v(s))g(w(s))\}\\
        &\quad\times f_{1}(y_{1}^{(1)},x+h_{1}t)f_{2}(y_{1}^{(2)},x+h_{2}v)f_{1}(x+h_{1}u)f_{2}(x+h_{2}w)dy_{1}^{(1)}dy_{1}^{(2)}dtdudvdw\\
        &=\int\E\{\rho(Y^{(1)},Y^{(2)})\mid X^{(1)}=x+h_{1}t,X^{(2)}=x+h_{2}v\}\prod_{s=1}^{p}\{g(t(s))g(u(s))g(v(s))g(w(s))\}\\
        &\quad\times f_{1}(x+h_{1}t)f_{1}(x+h_{1}u)f_{2}(x+h_{2}v)f_{2}(x+h_{2}w)dtdudvdw\\
        &\overset{(ii)}{=}\int\E\{\rho(Y^{(1)},Y^{(2)})\mid X^{(1)}=x,X^{(2)}=x\}\prod_{s=1}^{p}\{g(t(s))g(u(s))g(v(s))g(w(s))\}\\
        &\quad\times \{f_{1}(x)f_{2}(x)\}^{2}dtdudvdw+O(h_{1}^{\nu}+h_{2}^{\nu})\\
        &=\E\{\rho(Y^{(1)},Y^{(2)})\mid X^{(1)}=x,X^{(2)}=x\}\{f_{1}(x)f_{2}(x)\}^{2}+O(h_{1}^{\nu}+h_{2}^{\nu}),
    \end{align*}
    where $(i)$ follows from change of variables, and $(ii)$ holds by Taylor's theorem and Assumptions \ref{assum1} and \ref{assum3}-\ref{assum4}. Similarly, one can verify that
    \begin{align*}
        &\quad\E\{\rho(Y_{1}^{(1)},Y_{2}^{(1)})G_{h_{1}}(X_{1}^{(1)}-x)G_{h_{1}}(X_{2}^{(1)}-x)G_{h_{2}}(X_{1}^{(2)}-x)G_{h_{2}}(X_{2}^{(2)}-x)\}\\
        &=\E\{\rho(Y^{(1)},Y^{(1)\prime})\mid X^{(1)}=x,X^{(1)\prime}=x\}\{f_{1}(x)f_{2}(x)\}^{2}+O(h_{1}^{\nu}+h_{2}^{\nu}),\\
        &\quad\E\{\rho(Y_{1}^{(2)},Y_{2}^{(2)})G_{h_{1}}(X_{1}^{(1)}-x)G_{h_{1}}(X_{2}^{(1)}-x)G_{h_{2}}(X_{1}^{(2)}-x)G_{h_{2}}(X_{2}^{(2)}-x)\}\\
        &=\E\{\rho(Y^{(2)},Y^{(2)\prime})\mid X^{(2)}=x,X^{(2)\prime}=x\}\{f_{1}(x)f_{2}(x)\}^{2}+O(h_{1}^{\nu}+h_{2}^{\nu}),\\
        &\quad\E\{G_{h_{1}}(X_{1}^{(1)}-x)G_{h_{1}}(X_{2}^{(1)}-x)G_{h_{2}}(X_{1}^{(2)}-x)G_{h_{2}}(X_{2}^{(2)}-x)\}\\
        &=\{f_{1}(x)f_{2}(x)\}^{2}+O(h_{1}^{\nu}+h_{2}^{\nu}). 
    \end{align*}
    Thus, by Assumption \ref{assum2},
    \[
    \E(U_{n_{1},n_{2}})=O(h_{1}^{\nu}+h_{2}^{\nu})=o(1).
    \]

    For the second part of (\ref{thm2_pf}), following Definition \ref{defa2} and Theorem \ref{thma1} ($m_{1}=m_{2}=2$), we have
    \[
    \var(U_{n_{1},n_{2}})=\binom{n_{1}}{2}^{-1}\binom{n_{2}}{2}^{-1}\sum_{c=0}^{2}\sum_{d=0}^{2}\binom{2}{c}\binom{2}{d}\binom{n_{1}-2}{2-c}\binom{n_{2}-2}{2-d}\sigma_{c,d}^{2}.
    \]
    We first calculate
    \[
    \sigma_{2,2}^{2}=\var\{\psi(Z_{1}^{(1)},Z_{2}^{(1)};Z_{1}^{(2)},Z_{2}^{(2)})\}=\E\{\psi^{2}(Z_{1}^{(1)},Z_{2}^{(1)};Z_{1}^{(2)},Z_{2}^{(2)})\}-\{\E(U_{n_{1},n_{2}})\}^{2}.
    \]
    As for $\E\{\psi^{2}(Z_{1}^{(1)},Z_{2}^{(1)};Z_{1}^{(2)},Z_{2}^{(2)})\}$, it can be expanded into several terms, and each of these terms can be shown to be of order $(h_{1}h_{2})^{-2p}$. We only present the proof for the term below. Derivations for the other terms are similar. Note that
    \begin{align*}
        &\quad\E\Big[\Big\{\rho(Y_{1}^{(1)},Y_{1}^{(2)})G_{h_{1}}(X_{1}^{(1)}-x)G_{h_{1}}(X_{2}^{(1)}-x)G_{h_{2}}(X_{1}^{(2)}-x)G_{h_{2}}(X_{2}^{(2)}-x)\Big\}^{2}\Big]\\
        &=\int\left\{\rho(y_{1}^{(1)},y_{1}^{(2)})G_{h_{1}}(x_{1}^{(1)}-x)G_{h_{1}}(x_{2}^{(1)}-x)G_{h_{2}}(x_{1}^{(2)}-x)G_{h_{2}}(x_{2}^{(2)}-x)\right\}^{2}\\
        &\quad\times f_{1}(y_{1}^{(1)},x_{1}^{(1)})f_{2}(y_{1}^{(2)},x_{1}^{(2)})f_{1}(x_{2}^{(1)})f_{2}(x_{2}^{(2)})dy_{1}^{(1)}dy_{1}^{(2)}dx_{1}^{(1)}dx_{2}^{(1)}dx_{1}^{(2)}dx_{2}^{(2)}\\
        &\overset{(i)}{=}(h_{1}h_{2})^{-2p}\int\bigg[\rho(y_{1}^{(1)},y_{1}^{(2)})\prod_{s=1}^{p}\{g(t(s))g(u(s))g(v(s))g(w(s))\}\bigg]^{2}\\
        &\quad\times f_{1}(y_{1}^{(1)},x+h_{1}t)f_{2}(y_{1}^{(2)},x+h_{2}v)f_{1}(x+h_{1}u)f_{2}(x+h_{2}w)dy_{1}^{(1)}dy_{1}^{(2)}dtdudvdw\\
        &\asymp(h_{1}h_{2})^{-2p},
    \end{align*}
    where $(i)$ follows from the same change of variables technique as in the proof of the first part of (\ref{thm2_pf}) above. Hence, $\sigma_{2,2}^{2}=O((h_{1}h_{2})^{-2p})$. 
    
    Analogously, using change of variables, one can verify that $\sigma_{1,0}^{2}=O(h_{1}^{-p})$, $\sigma_{0,1}^{2}=O(h_{2}^{-p})$, $\sigma_{2,0}^{2}=O(h_{1}^{-2p})$, $\sigma_{1,1}^{2}=O((h_{1}h_{2})^{-p})$, $\sigma_{0,2}^{2}=O(h_{2}^{-2p})$, $\sigma_{2,1}^{2}=O((h_{1}^{2}h_{2})^{-p})$ and $\sigma_{1,2}^{2}=O((h_{1}h_{2}^{2})^{-p})$. Therefore, by Assumption \ref{assum2}, 
    \[
    \var(U_{n_{1},n_{2}})=\sum_{c=0}^{2}\sum_{d=0}^{2}O(n_{1}^{-c}n_{2}^{-d})\sigma_{c,d}^{2}=O((n_{1}h_{1}^{p})^{-1}+(n_{2}h_{2}^{p})^{-1})=o(1).
    \]
    This completes the proof. 
\end{proof}

\begin{proof}[Proof of Theorem \ref{thm3}]
    Recall from the proof of Theorem \ref{thm2} that
    \[
    \widehat{\D}_{\rho}(x)-\D_{\rho}(x)=\big[\{f_{1}(x)f_{2}(x)\}^{2}+o_{p}(1)\big]^{-1}U_{n_{1},n_{2}}.
    \]
    Following Definition \ref{defa3} and Theorem \ref{thma2}, we have the Hoeffding decomposition for the U-statistic $U_{n_{1},n_{2}}$:
    \[
    U_{n_{1},n_{2}}=\sum_{c=0}^{2}\sum_{d=0}^{2}\binom{2}{c}\binom{2}{d}H_{n_{1},n_{2}}^{(c,d)}=\E(U_{n_{1},n_{2}})+2H_{n_{1},n_{2}}^{(1,0)}+2H_{n_{1},n_{2}}^{(0,1)}+R_{n_{1},n_{2}},
    \]
    where 
    \[
    H_{n_{1},n_{2}}^{(1,0)}=\frac{1}{n_{1}}\sum_{i=1}^{n_{1}}\phi^{(1,0)}(Z_{i}^{(1)}),\quad H_{n_{1},n_{2}}^{(0,1)}=\frac{1}{n_{2}}\sum_{j=1}^{n_{2}}\phi^{(0,1)}(Z_{j}^{(2)}),
    \]
    and $R_{n_{1},n_{2}}$ is the remainder term. Furthermore, we give the explicit expression of $\phi^{(1,0)}(z_{1}^{(1)})$ (the expression of $\phi^{(0,1)}(z_{1}^{(2)})$ is similar and omitted): 
    \begin{align}
        \begin{split}\label{eq:hoeff1}
            \phi^{(1,0)}(z_{1}^{(1)})&=\E\Big\{\rho(y_{1}^{(1)},Y_{1}^{(2)})G_{h_{1}}(x_{1}^{(1)}-x)G_{h_{1}}(X_{2}^{(1)}-x)G_{h_{2}}(X_{1}^{(2)}-x)G_{h_{2}}(X_{2}^{(2)}-x)\Big\}\\
            &\quad+\E\Big\{\rho(Y_{2}^{(1)},Y_{1}^{(2)})G_{h_{1}}(x_{1}^{(1)}-x)G_{h_{1}}(X_{2}^{(1)}-x)G_{h_{2}}(X_{1}^{(2)}-x)G_{h_{2}}(X_{2}^{(2)}-x)\Big\}\\
            &\quad-\E\Big\{\rho(y_{1}^{(1)},Y_{2}^{(1)})G_{h_{1}}(x_{1}^{(1)}-x)G_{h_{1}}(X_{2}^{(1)}-x)G_{h_{2}}(X_{1}^{(2)}-x)G_{h_{2}}(X_{2}^{(2)}-x)\Big\}\\
            &\quad-\E\Big\{\rho(Y_{1}^{(2)},Y_{2}^{(2)})G_{h_{1}}(x_{1}^{(1)}-x)G_{h_{1}}(X_{2}^{(1)}-x)G_{h_{2}}(X_{1}^{(2)}-x)G_{h_{2}}(X_{2}^{(2)}-x)\Big\}\\
            &\quad-\D_{\rho}(x)\E\Big\{G_{h_{1}}(x_{1}^{(1)}-x)G_{h_{1}}(X_{2}^{(1)}-x)G_{h_{2}}(X_{1}^{(2)}-x)G_{h_{2}}(X_{2}^{(2)}-x)\Big\}\\
            &\quad-\E(U_{n_{1},n_{2}}).
        \end{split}
    \end{align}

    By similar calculations as in the proof of Theorem \ref{thm2}, we have
    \begin{align}
        \xi_{1}^{2}&\coloneqq\var\{\phi^{(1,0)}(Z_{1}^{(1)})\}=\sigma_{1,0}^{2}\asymp h_{1}^{-p}, \label{eq-xi1}\\
        \xi_{2}^{2}&\coloneqq\var\{\phi^{(0,1)}(Z_{1}^{(2)})\}=\sigma_{0,1}^{2}\asymp h_{2}^{-p}, \label{eq-xi2}
    \end{align}
    under $H_{a}$ in (\ref{eq:hypo2}). Applying Lyapunov CLT as in Section 3.4 of \citet{ullah1999nonparametric},
    \[
    \frac{H_{n_{1},n_{2}}^{(1,0)}+H_{n_{1},n_{2}}^{(0,1)}}{\sqrt{\frac{\xi_{1}^{2}}{n_{1}}+\frac{\xi_{2}^{2}}{n_{2}}}}\convd\N(0,1).
    \]
    The condition that $\int|g(u)|^{2+\delta}du<\infty$ for some $\delta>0$ used in \citet{ullah1999nonparametric} is implied by our Assumption \ref{assum1} that $g$ is bounded and $\int g^{2}(u)du<\infty$.
    
    Also, one can show that
    \begin{align*}
        \var(H_{n_{1},n_{2}}^{(2,0)})&\asymp n_{1}^{-2}\delta_{2,0}^{2}\asymp n_{1}^{-2}\sigma_{2,0}^{2}\asymp (n_{1}h_{1}^{p})^{-2},\\
        \var(H_{n_{1},n_{2}}^{(1,1)})&\asymp (n_{1}n_{2})^{-1}\delta_{1,1}^{2}\asymp (n_{1}n_{2})^{-1}\sigma_{1,1}^{2}\asymp (n_{1}h_{1}^{p})^{-1}(n_{2}h_{2}^{p})^{-1},\\
        \var(H_{n_{1},n_{2}}^{(0,2)})&\asymp n_{2}^{-2}\delta_{0,2}^{2}\asymp n_{2}^{-2}\sigma_{0,2}^{2}\asymp (n_{2}h_{2}^{p})^{-2},\\
        \var(H_{n_{1},n_{2}}^{(2,1)})&\asymp n_{1}^{-2}n_{2}^{-1}\delta_{2,1}^{2}\asymp n_{1}^{-2}n_{2}^{-1}\sigma_{2,1}^{2}\asymp (n_{1}h_{1}^{p})^{-2}(n_{2}h_{2}^{p})^{-1},\\
        \var(H_{n_{1},n_{2}}^{(1,2)})&\asymp n_{1}^{-1}n_{2}^{-2}\delta_{1,2}^{2}\asymp n_{1}^{-1}n_{2}^{-2}\sigma_{1,2}^{2}\asymp (n_{1}h_{1}^{p})^{-1}(n_{2}h_{2}^{p})^{-2},\\
        \var(H_{n_{1},n_{2}}^{(2,2)})&\asymp (n_{1}n_{2})^{-2}\delta_{2,2}^{2}\asymp (n_{1}n_{2})^{-2}\sigma_{2,2}^{2}\asymp (n_{1}h_{1}^{p})^{-2}(n_{2}h_{2}^{p})^{-2}.
    \end{align*}
    Thus, 
    \[
    \frac{\var(R_{n_{1},n_{2}})}{\frac{\xi_{1}^{2}}{n_{1}}+\frac{\xi_{2}^{2}}{n_{2}}}=O((n_{1}h_{1}^{p})^{-1}+(n_{2}h_{2}^{p})^{-1}).
    \]
    
    Now we have
    \begin{align*}
        \frac{U_{n_{1},n_{2}}}{\sqrt{\frac{\xi_{1}^{2}}{n_{1}}+\frac{\xi_{2}^{2}}{n_{2}}}}&=\frac{2H_{n_{1},n_{2}}^{(1,0)}+2H_{n_{1},n_{2}}^{(0,1)}}{\sqrt{\frac{\xi_{1}^{2}}{n_{1}}+\frac{\xi_{2}^{2}}{n_{2}}}}+\frac{\E(U_{n_{1},n_{2}})}{\sqrt{\frac{\xi_{1}^{2}}{n_{1}}+\frac{\xi_{2}^{2}}{n_{2}}}}+\frac{R_{n_{1},n_{2}}}{\sqrt{\frac{\xi_{1}^{2}}{n_{1}}+\frac{\xi_{2}^{2}}{n_{2}}}}\\
        &=\frac{2H_{n_{1},n_{2}}^{(1,0)}+2H_{n_{1},n_{2}}^{(0,1)}}{\sqrt{\frac{\xi_{1}^{2}}{n_{1}}+\frac{\xi_{2}^{2}}{n_{2}}}}+O((n_{1}h_{1}^{p})^{1/2}h_{1}^{\nu}+(n_{2}h_{2}^{p})^{1/2}h_{2}^{\nu})+O_{p}((n_{1}h_{1}^{p})^{-1/2}+(n_{2}h_{2}^{p})^{-1/2})\\
        &=\frac{2H_{n_{1},n_{2}}^{(1,0)}+2H_{n_{1},n_{2}}^{(0,1)}}{\sqrt{\frac{\xi_{1}^{2}}{n_{1}}+\frac{\xi_{2}^{2}}{n_{2}}}}+o_{p}(1)\convd\N(0,4),
    \end{align*}
    where we have used Assumption \ref{assum2} and $n_{l}^{1/2}h_{l}^{p/2+\nu}\rightarrow 0$ as $n_{l}\rightarrow\infty$ for $l=1,2$. 
    Then
    \[
    \frac{\widehat{\D}_{\rho}(x)-\D_{\rho}(x)}{\sqrt{\frac{\xi_{1}^{2}}{n_{1}}+\frac{\xi_{2}^{2}}{n_{2}}}}=\frac{U_{n_{1},n_{2}}}{\big[\{f_{1}(x)f_{2}(x)\}^{2}+o_{p}(1)\big]\sqrt{\frac{\xi_{1}^{2}}{n_{1}}+\frac{\xi_{2}^{2}}{n_{2}}}}\convd\N(0,4\{f_{1}(x)f_{2}(x)\}^{-4}).
    \]
\end{proof}

\begin{proof}[Proof of Theorem \ref{thm4}]
    As $\D_{\rho}(x)=0$ under $H_{0}$ in (\ref{eq:hypo2}), following the proof of Theorems \ref{thm2}-\ref{thm3}, we have
    \[
    \widehat{\D}_{\rho}(x)=\big[\{f_{1}(x)f_{2}(x)\}^{2}+o_{p}(1)\big]^{-1}U_{n_{1},n_{2}},
    \]
    and the Hoeffding decomposition for $U_{n_{1},n_{2}}$:
    \[
    U_{n_{1},n_{2}}=\E(U_{n_{1},n_{2}})+2H_{n_{1},n_{2}}^{(1,0)}+2H_{n_{1},n_{2}}^{(0,1)}+H_{n_{1},n_{2}}^{(2,0)}+4H_{n_{1},n_{2}}^{(1,1)}+H_{n_{1},n_{2}}^{(0,2)}+R_{n_{1},n_{2}},
    \]
    where
    \begin{align*}
        H_{n_{1},n_{2}}^{(2,0)}&=\binom{n_{1}}{2}^{-1}\sum_{1\le i_{1}<i_{2}\le n_{1}}\phi^{(2,0)}(Z_{i_{1}}^{(1)},Z_{i_{2}}^{(1)}),\\
        H_{n_{1},n_{2}}^{(1,1)}&=\frac{1}{n_{1}n_{2}}\sum_{i=1}^{n_{1}}\sum_{j=1}^{n_{2}}\phi^{(1,1)}(Z_{i}^{(1)};Z_{j}^{(2)}),\\
        H_{n_{1},n_{2}}^{(0,2)}&=\binom{n_{2}}{2}^{-1}\sum_{1\le j_{1}<j_{2}\le n_{2}}\phi^{(0,2)}(Z_{j_{1}}^{(2)},Z_{j_{2}}^{(2)}),
    \end{align*}
    and $R_{n_{1},n_{2}}$ is the remainder term, with a slight abuse of notation.
    
    Under $H_{0}$ in (\ref{eq:hypo2}), one can verify that $\phi^{(1,0)}(z_{1}^{(1)})\neq 0$ and $\phi^{(0,1)}(z_{1}^{(2)})\neq 0$, and thus the U-statistic $U_{n_{1},n_{2}}$ is nondegenerate. Nonetheless, a more in-depth analysis reveals that $H_{n_{1},n_{2}}^{(1,0)}$ and $H_{n_{1},n_{2}}^{(0,1)}$ can be asymptotically negligible under $H_{0}$ in (\ref{eq:hypo2}). Specifically, by similar calculations as in the proof of Theorem \ref{thm2}, 
    \begin{align*}
        &\quad\E\Big\{\rho(y_{1}^{(1)},Y_{1}^{(2)})G_{h_{1}}(x_{1}^{(1)}-x)G_{h_{1}}(X_{2}^{(1)}-x)G_{h_{2}}(X_{1}^{(2)}-x)G_{h_{2}}(X_{2}^{(2)}-x)\Big\}\\
        &=G_{h_{1}}(x_{1}^{(1)}-x)\int\rho(y_{1}^{(1)},y_{1}^{(2)})f_{2}(y_{1}^{(2)}\mid x)dy_{1}^{(2)}\times\big[f_{1}(x)\{f_{2}(x)\}^{2}+O(h_{1}^{\nu}+h_{2}^{\nu})\big],\\
        &\quad\E\Big\{\rho(Y_{2}^{(1)},Y_{1}^{(2)})G_{h_{1}}(x_{1}^{(1)}-x)G_{h_{1}}(X_{2}^{(1)}-x)G_{h_{2}}(X_{1}^{(2)}-x)G_{h_{2}}(X_{2}^{(2)}-x)\Big\}\\
        &=G_{h_{1}}(x_{1}^{(1)}-x)\bigg[\int\rho(y_{2}^{(1)},y_{1}^{(2)})f_{1}(y_{2}^{(1)}\mid x)f_{2}(y_{1}^{(2)}\mid x)dy_{2}^{(1)}dy_{1}^{(2)}\times f_{1}(x)\{f_{2}(x)\}^{2}+O(h_{1}^{\nu}+h_{2}^{\nu})\bigg],\\
        &\quad\E\Big\{\rho(y_{1}^{(1)},Y_{2}^{(1)})G_{h_{1}}(x_{1}^{(1)}-x)G_{h_{1}}(X_{2}^{(1)}-x)G_{h_{2}}(X_{1}^{(2)}-x)G_{h_{2}}(X_{2}^{(2)}-x)\Big\}\\
        &=G_{h_{1}}(x_{1}^{(1)}-x)\int\rho(y_{1}^{(1)},y_{2}^{(1)})f_{1}(y_{2}^{(1)}\mid x)dy_{2}^{(1)}\times\big[f_{1}(x)\{f_{2}(x)\}^{2}+O(h_{1}^{\nu}+h_{2}^{\nu})\big],\\
        &\quad\E\Big\{\rho(Y_{1}^{(2)},Y_{2}^{(2)})G_{h_{1}}(x_{1}^{(1)}-x)G_{h_{1}}(X_{2}^{(1)}-x)G_{h_{2}}(X_{1}^{(2)}-x)G_{h_{2}}(X_{2}^{(2)}-x)\Big\}\\
        &=G_{h_{1}}(x_{1}^{(1)}-x)\bigg[\int\rho(y_{1}^{(2)},y_{2}^{(2)})f_{2}(y_{1}^{(2)}\mid x)f_{2}(y_{2}^{(2)}\mid x)dy_{1}^{(2)}dy_{2}^{(2)}\times f_{1}(x)\{f_{2}(x)\}^{2}+O(h_{1}^{\nu}+h_{2}^{\nu})\bigg].           
    \end{align*}
    Under $H_{0}$ in (\ref{eq:hypo2}), we have $P_{Y\mid X=x}\coloneqq P_{Y\mid X=x}^{(1)}=P_{Y\mid X=x}^{(2)}$, and $f(y\mid x)\coloneqq f_{1}(y\mid x)=f_{2}(y\mid x)$ for all $y$ at the fixed $x$. Hence, for $\phi^{(1,0)}(z_{1}^{(1)})$ in (\ref{eq:hoeff1}), we have
    \[
    \phi^{(1,0)}(z_{1}^{(1)})=G_{h_{1}}(x_{1}^{(1)}-x)\bigg\{\int\rho(y_{1}^{(1)},y)f(y\mid x)dy+1\bigg\}\times O(h_{1}^{\nu}+h_{2}^{\nu})+O(h_{1}^{\nu}+h_{2}^{\nu}),
    \]
    under $H_{0}$ in (\ref{eq:hypo2}). It implies that
    \[
    \var\{\phi^{(1,0)}(Z_{1}^{(1)})\}=O(h_{1}^{-p}(h_{1}^{2\nu}+h_{2}^{2\nu})),
    \]
    and thus, 
    \[
    \var(H_{n_{1},n_{2}}^{(1,0)})=O(n_{1}^{-1}h_{1}^{-p}(h_{1}^{2\nu}+h_{2}^{2\nu})).
    \]
    Analogously, one can show that $\var(H_{n_{1},n_{2}}^{(0,1)})=O(n_{2}^{-1}h_{2}^{-p}(h_{1}^{2\nu}+h_{2}^{2\nu}))$ under $H_{0}$ in (\ref{eq:hypo2}). Following the proof of Theorem \ref{thm3}, we have $\var(H_{n_{1},n_{2}}^{(2,0)})\asymp (n_{1}h_{1}^{p})^{-2}$, $\var(H_{n_{1},n_{2}}^{(1,1)})\asymp (n_{1}h_{1}^{p})^{-1}(n_{2}h_{2}^{p})^{-1}$, and $\var(H_{n_{1},n_{2}}^{(0,2)})\asymp (n_{2}h_{2}^{p})^{-2}$. Hence, when undersmoothing is employed, the asymptotic distribution of $U_{n_{1},n_{2}}$ is determined by $H_{n_{1},n_{2}}^{(2,0)}+4H_{n_{1},n_{2}}^{(1,1)}+H_{n_{1},n_{2}}^{(0,2)}$, instead of $2H_{n_{1},n_{2}}^{(1,0)}+2H_{n_{1},n_{2}}^{(0,1)}$. 
    
    The proof below is essentially similar to that in Appendix B.1 of \citet{gretton2012kernel}. The main difference lies in that the spectral decomposition is performed with respect to the conditional distribution. Define the double centered version of $\rho$:
    \[
    \rhotilde(y,y')=\rho(y,y')-\E\{\rho(y,Y')\mid X'=x\}-\E\{\rho(Y,y')\mid X=x\}+\E\{\rho(Y,Y')\mid X=x,X'=x\},
    \]
    where $Y\mid X=x,Y'\mid X'=x\overset{\iid}{\sim}P_{Y\mid X=x}$. As $P_{Y\mid X=x}\in\M_{\rho}^{2+\delta}(\Y)$, the semimetric $\rhotilde$ is square integrable with respect to $P_{Y\mid X=x}$. Then $\rhotilde(y,y')$ admits a spectral decomposition
    \begin{equation}\label{eq:spd}
        \rhotilde(y,y')=\sum_{r=1}^{\infty}\lambda_{r}\phi_{r}(y)\phi_{r}(y'),
    \end{equation}
    where $\{\lambda_{r}\}_{r=1}^{\infty}$ are the eigenvalues and $\{\phi_{r}(\cdot)\}_{r=1}^{\infty}$ are the corresponding orthonormal eigenfunctions of $\rhotilde$ with respect to the conditional distribution $P_{Y\mid X=x}$, i.e., 
    \begin{align*}
        \E\{\rhotilde(y,Y')\phi_{r}(Y')\mid X'=x\}&=\lambda_{r}\phi_{r}(y),\\
        \E\{\phi_{r_{1}}(Y)\phi_{r_{2}}(Y)\mid X=x\}&=\left\{
            \begin{aligned}
                1,\quad\text{if }r_{1}=r_{2},\\
                0,\quad\text{if }r_{1}\neq r_{2}.
            \end{aligned}\right.
    \end{align*}
    
    We now find the asymptotic distribution of $H_{n_{1},n_{2}}^{(2,0)}+4H_{n_{1},n_{2}}^{(1,1)}+H_{n_{1},n_{2}}^{(0,2)}$. First, under $H_{0}$ in (\ref{eq:hypo2}),
    \begin{align*}
        \phi^{(1,1)}(z_{1}^{(1)};z_{1}^{(2)})&=\frac{1}{2}f_{1}(x)f_{2}(x)\times\Big[\rhotilde(y_{1}^{(1)},y_{1}^{(2)})G_{h_{1}}(x_{1}^{(1)}-x)G_{h_{2}}(x_{1}^{(2)}-x)\\
        &\quad-\E\Big\{\rhotilde(y_{1}^{(1)},Y_{1}^{(2)})G_{h_{1}}(x_{1}^{(1)}-x)G_{h_{2}}(X_{1}^{(2)}-x)\Big\}\\
        &\quad-\E\Big\{\rhotilde(Y_{1}^{(1)},y_{1}^{(2)})G_{h_{1}}(X_{1}^{(1)}-x)G_{h_{2}}(x_{1}^{(2)}-x)\Big\}\\
        &\quad+\E\Big\{\rhotilde(Y_{1}^{(1)},Y_{1}^{(2)})G_{h_{1}}(X_{1}^{(1)}-x)G_{h_{2}}(X_{1}^{(2)}-x)\Big\}\Big]\\
        &\quad\times\{1+O(h_{1}^{\nu}+h_{2}^{\nu})\}.
    \end{align*}
    Using the spectral decomposition (\ref{eq:spd}), we have
    \begin{align*}
        H_{n_{1},n_{2}}^{(1,1)}&=\frac{1}{2}f_{1}(x)f_{2}(x)\sum_{r=1}^{\infty}\lambda_{r}\times\frac{1}{n_{1}}\sum_{i=1}^{n_{1}}\Big[\phi_{r}(Y_{i}^{(1)})G_{h_{1}}(X_{i}^{(1)}-x)-\E\Big\{\phi_{r}(Y^{(1)})G_{h_{1}}(X^{(1)}-x)\Big\}\Big]\\
        &\quad\times\frac{1}{n_{2}}\sum_{j=1}^{n_{2}}\Big[\phi_{r}(Y_{j}^{(2)})G_{h_{2}}(X_{j}^{(2)}-x)-\E\Big\{\phi_{r}(Y^{(2)})G_{h_{2}}(X^{(2)}-x)\Big\}\Big]\\
        &\quad+O_{p}((n_{1}h_{1}^{p}n_{2}h_{2}^{p})^{-1/2}(h_{1}^{\nu}+h_{2}^{\nu})).
    \end{align*}
    Applying Lyapunov CLT, 
    \begin{align*}
        \sqrt{n_{1}h_{1}^{p}}\frac{1}{n_{1}}\sum_{i=1}^{n_{1}}\Big[\phi_{r}(Y_{i}^{(1)})G_{h_{1}}(X_{i}^{(1)}-x)-\E\Big\{\phi_{r}(Y^{(1)})G_{h_{1}}(X^{(1)}-x)\Big\}\Big]\\
        \convd\N\bigg(0,f_{1}(x)\bigg\{\int g^{2}(u)du\bigg\}^{p}\bigg),\\
        \sqrt{n_{2}h_{2}^{p}}\frac{1}{n_{2}}\sum_{j=1}^{n_{2}}\Big[\phi_{r}(Y_{j}^{(2)})G_{h_{2}}(X_{j}^{(2)}-x)-\E\Big\{\phi_{r}(Y^{(2)})G_{h_{2}}(X^{(2)}-x)\Big\}\Big]\\
        \convd\N\bigg(0,f_{2}(x)\bigg\{\int g^{2}(u)du\bigg\}^{p}\bigg),
    \end{align*}
    since $\E\{\phi_{r}^{2}(Y)\mid X=x\}=1$. Also, for $l=1,2$ and $r_{1}\neq r_{2}$, we have
    \begin{align*}
        &\quad h_{l}^{p}\cov\Big[\phi_{r_{1}}(Y^{(l)})G_{h_{l}}(X^{(l)}-x),\phi_{r_{2}}(Y^{(l)})G_{h_{l}}(X^{(l)}-x)\Big]\\
        &\longrightarrow\E\{\phi_{r_{1}}(Y)\phi_{r_{2}}(Y)\mid X=x\}f_{l}(x)\bigg\{\int g^{2}(u)du\bigg\}^{p}=0.
    \end{align*}
    Thus, 
    \[
    \{f_{1}(x)f_{2}(x)\}^{-2}\sqrt{n_{1}h_{1}^{p}n_{2}h_{2}^{p}}H_{n_{1},n_{2}}^{(1,1)}\convd\frac{1}{2}\sum^{\infty}_{r=1}\lambda_{r}\zeta_{r}\eta_{r},
    \]
    where $\{\zeta_r\}_{r=1}^{\infty}$ and $\{\eta_r\}_{r=1}^{\infty}$ are two sequences of independent normal random variables with
    \begin{align*}
        \zeta_{r}&\overset{\iid}{\sim}\N(0,\sigma_{1}^{2}),\quad\sigma_{1}^{2}=\frac{1}{f_{1}(x)}\bigg\{\int g^{2}(u)du\bigg\}^{p},\\
        \eta_{r}&\overset{\iid}{\sim}\N(0,\sigma_{2}^{2}),\quad\sigma_{2}^{2}=\frac{1}{f_{2}(x)}\bigg\{\int g^{2}(u)du\bigg\}^{p}.
    \end{align*}

    Similarly, one can show that
    \[
    \{f_{1}(x)f_{2}(x)\}^{-2}n_{1}h_{1}^{p}H_{n_{1},n_{2}}^{(2,0)}\convd-\sum_{r=1}^{\infty}\lambda_{r}(\zeta_{r}^{2}-\sigma_{1}^{2}),
    \]
    and
    \[
    \{f_{1}(x)f_{2}(x)\}^{-2}n_{2}h_{2}^{p}H_{n_{1},n_{2}}^{(0,2)}\convd-\sum_{r=1}^{\infty}\lambda_{r}(\eta_{r}^{2}-\sigma_{2}^{2}).
    \]
    Combining the above results, we have
    \begin{align*}
        &\quad\{f_{1}(x)f_{2}(x)\}^{-2}(n_{1}h_{1}^{p}+n_{2}h_{2}^{p})(H_{n_{1},n_{2}}^{(2,0)}+4H_{n_{1},n_{2}}^{(1,1)}+H_{n_{1},n_{2}}^{(0,2)})\\
        &\convd-\frac{1}{\tau}\sum_{r=1}^{\infty}\lambda_{r}(\zeta_{r}^{2}-\sigma_{1}^{2})+\frac{2}{\sqrt{\tau(1-\tau)}}\sum^{\infty}_{r=1}\lambda_{r}\zeta_{r}\eta_{r}-\frac{1}{1-\tau}\sum_{r=1}^{\infty}\lambda_{r}(\eta_{r}^{2}-\sigma_{2}^{2})\\
        &=-\sum_{r=1}^{\infty}\lambda_{r}\bigg\{\bigg(\frac{1}{\sqrt{\tau}}\zeta_{r}-\frac{1}{\sqrt{1-\tau}}\eta_{r}\bigg)^{2}-\frac{(1-\tau)\sigma_{1}^{2}+\tau\sigma_{2}^{2}}{\tau(1-\tau)}\bigg\}.
    \end{align*}
    Also, according to the results in the proof of Theorem \ref{thm3},
    \[
    (n_{1}h_{1}^{p}+n_{2}h_{2}^{p})^{2}\var(R_{n_{1},n_{2}})=O((n_{1}h_{1}^{p})^{-1}+(n_{2}h_{2}^{p})^{-1}).
    \]
    Therefore,
    \begin{align*}
        &\quad\{f_{1}(x)f_{2}(x)\}^{-2}(n_{1}h_{1}^{p}+n_{2}h_{2}^{p})U_{n_{1},n_{2}}\\
        &=\{f_{1}(x)f_{2}(x)\}^{-2}(n_{1}h_{1}^{p}+n_{2}h_{2}^{p})(H_{n_{1},n_{2}}^{(2,0)}+4H_{n_{1},n_{2}}^{(1,1)}+H_{n_{1},n_{2}}^{(0,2)})\\
        &\quad+\{f_{1}(x)f_{2}(x)\}^{-2}(n_{1}h_{1}^{p}+n_{2}h_{2}^{p})\{\E(U_{n_{1},n_{2}})+2H_{n_{1},n_{2}}^{(1,0)}+2H_{n_{1},n_{2}}^{(0,1)}+R_{n_{1},n_{2}}\}\\
        &=\{f_{1}(x)f_{2}(x)\}^{-2}(n_{1}h_{1}^{p}+n_{2}h_{2}^{p})(H_{n_{1},n_{2}}^{(2,0)}+4H_{n_{1},n_{2}}^{(1,1)}+H_{n_{1},n_{2}}^{(0,2)})\\
        &\quad+O(n_{1}h_{1}^{p+\nu}+n_{2}h_{2}^{p+\nu})+O_{p}(n_{1}^{1/2}h_{1}^{p/2+\nu}+n_{2}^{1/2}h_{2}^{p/2+\nu})+O_{p}((n_{1}h_{1}^{p})^{-1/2}+(n_{2}h_{2}^{p})^{-1/2})\\
        &=\{f_{1}(x)f_{2}(x)\}^{-2}(n_{1}h_{1}^{p}+n_{2}h_{2}^{p})(H_{n_{1},n_{2}}^{(2,0)}+4H_{n_{1},n_{2}}^{(1,1)}+H_{n_{1},n_{2}}^{(0,2)})+o_{p}(1)\\
        &\convd-\sum_{r=1}^{\infty}\lambda_{r}\bigg\{\bigg(\frac{1}{\sqrt{\tau}}\zeta_{r}-\frac{1}{\sqrt{1-\tau}}\eta_{r}\bigg)^{2}-\frac{(1-\tau)\sigma_{1}^{2}+\tau\sigma_{2}^{2}}{\tau(1-\tau)}\bigg\},
    \end{align*}
    where we have used Assumption \ref{assum2}, and $n_{l}h_{l}^{p+\nu}\rightarrow 0$ as $n_{l}\rightarrow\infty$ for $l=1,2$. The result follows. 
\end{proof}

\begin{proof}[Proof of Theorem \ref{thm5}]
    Note that
    \[
    \widehat{\I}_{\rho}=\binom{n_{1}}{2}^{-1}\binom{n_{2}}{2}^{-1}\sum_{1\le i_{1}<i_{2}\le n_{1}}\sum_{1\le j_{1}<j_{2}\le n_{2}}\psi(Z_{i_{1}}^{(1)},Z_{i_{2}}^{(1)};Z_{j_{1}}^{(2)},Z_{j_{2}}^{(2)}),
    \]
    where
    \begin{align*}
        &\quad\psi(Z_{1}^{(1)},Z_{2}^{(1)};Z_{1}^{(2)},Z_{2}^{(2)})\\
        &=\frac{1}{4}\rho(Y_{1}^{(1)},Y_{1}^{(2)})\left\{G_{h_{1}}(X_{1}^{(1)}-X_{1}^{(2)})+G_{h_{2}}(X_{1}^{(2)}-X_{1}^{(1)})\right\}G_{h_{1}}(X_{2}^{(1)}-X_{1}^{(2)})G_{h_{2}}(X_{2}^{(2)}-X_{1}^{(1)})\\
        &\quad+\frac{1}{4}\rho(Y_{1}^{(1)},Y_{2}^{(2)})\left\{G_{h_{1}}(X_{1}^{(1)}-X_{2}^{(2)})+G_{h_{2}}(X_{2}^{(2)}-X_{1}^{(1)})\right\}G_{h_{1}}(X_{2}^{(1)}-X_{2}^{(2)})G_{h_{2}}(X_{1}^{(2)}-X_{1}^{(1)})\\
        &\quad+\frac{1}{4}\rho(Y_{2}^{(1)},Y_{1}^{(2)})\left\{G_{h_{1}}(X_{2}^{(1)}-X_{1}^{(2)})+G_{h_{2}}(X_{1}^{(2)}-X_{2}^{(1)})\right\}G_{h_{1}}(X_{1}^{(1)}-X_{1}^{(2)})G_{h_{2}}(X_{2}^{(2)}-X_{2}^{(1)})\\
        &\quad+\frac{1}{4}\rho(Y_{2}^{(1)},Y_{2}^{(2)})\left\{G_{h_{1}}(X_{2}^{(1)}-X_{2}^{(2)})+G_{h_{2}}(X_{2}^{(2)}-X_{2}^{(1)})\right\}G_{h_{1}}(X_{1}^{(1)}-X_{2}^{(2)})G_{h_{2}}(X_{1}^{(2)}-X_{2}^{(1)})\\
        &\quad-\frac{1}{2}\rho(Y_{1}^{(1)},Y_{2}^{(1)})G_{h_{1}}(X_{1}^{(1)}-X_{2}^{(1)})\\
        &\quad\times\left\{G_{h_{2}}(X_{1}^{(2)}-X_{1}^{(1)})G_{h_{2}}(X_{2}^{(2)}-X_{1}^{(1)})+G_{h_{2}}(X_{1}^{(2)}-X_{2}^{(1)})G_{h_{2}}(X_{2}^{(2)}-X_{2}^{(1)})\right\}\\
        &\quad-\frac{1}{2}\rho(Y_{1}^{(2)},Y_{2}^{(2)})G_{h_{2}}(X_{1}^{(2)}-X_{2}^{(2)})\\
        &\quad\times\left\{G_{h_{1}}(X_{1}^{(1)}-X_{1}^{(2)})G_{h_{1}}(X_{2}^{(1)}-X_{1}^{(2)})+G_{h_{1}}(X_{1}^{(1)}-X_{2}^{(2)})G_{h_{1}}(X_{2}^{(1)}-X_{2}^{(2)})\right\}.
    \end{align*}
    It suffices to show that
    \begin{equation}\label{thm5_pf}
        \E(\widehat{\I}_{\rho})=\I_{\rho}+o(1),\quad \var(\widehat{\I}_{\rho})=o(1).
    \end{equation}
    
    For the first part of (\ref{thm5_pf}), we have
    \begin{align*}
        \E(\widehat{\I}_{\rho})&=\E\{\rho(Y_{1}^{(1)},Y_{1}^{(2)})G_{h_{1}}(X_{1}^{(1)}-X_{1}^{(2)})G_{h_{1}}(X_{2}^{(1)}-X_{1}^{(2)})G_{h_{2}}(X_{2}^{(2)}-X_{1}^{(1)})\}\\
        &\quad+\E\{\rho(Y_{1}^{(1)},Y_{1}^{(2)})G_{h_{2}}(X_{1}^{(2)}-X_{1}^{(1)})G_{h_{1}}(X_{2}^{(1)}-X_{1}^{(2)})G_{h_{2}}(X_{2}^{(2)}-X_{1}^{(1)})\}\\
        &\quad-\E\{\rho(Y_{1}^{(1)},Y_{2}^{(1)})G_{h_{1}}(X_{1}^{(1)}-X_{2}^{(1)})G_{h_{2}}(X_{1}^{(2)}-X_{1}^{(1)})G_{h_{2}}(X_{2}^{(2)}-X_{1}^{(1)})\}\\
        &\quad-\E\{\rho(Y_{1}^{(2)},Y_{2}^{(2)})G_{h_{2}}(X_{1}^{(2)}-X_{2}^{(2)})G_{h_{1}}(X_{1}^{(1)}-X_{1}^{(2)})G_{h_{1}}(X_{2}^{(1)}-X_{1}^{(2)})\}.
    \end{align*}
    We first consider the term
    \begin{align*}
        &\quad\E\{\rho(Y_{1}^{(1)},Y_{1}^{(2)})G_{h_{1}}(X_{1}^{(1)}-X_{1}^{(2)})G_{h_{1}}(X_{2}^{(1)}-X_{1}^{(2)})G_{h_{2}}(X_{2}^{(2)}-X_{1}^{(1)})\}\\
        &=\int\rho(y_{1}^{(1)},y_{1}^{(2)})G_{h_{1}}(x_{1}^{(1)}-x_{1}^{(2)})G_{h_{1}}(x_{2}^{(1)}-x_{1}^{(2)})G_{h_{2}}(x_{2}^{(2)}-x_{1}^{(1)})\\
        &\quad\times f_{1}(y_{1}^{(1)},x_{1}^{(1)})f_{2}(y_{1}^{(2)},x_{1}^{(2)})f_{1}(x_{2}^{(1)})f_{2}(x_{2}^{(2)})dy_{1}^{(1)}dy_{1}^{(2)}dx_{1}^{(1)}dx_{2}^{(1)}dx_{1}^{(2)}dx_{2}^{(2)}.
    \end{align*}
    Let $x=x_{1}^{(2)}$, $u=h_{1}^{-1}(x_{1}^{(1)}-x_{1}^{(2)})$, $v=h_{1}^{-1}(x_{2}^{(1)}-x_{1}^{(2)})$ and $w=h_{2}^{-1}(x_{2}^{(2)}-x_{1}^{(1)})$. We have
    \begin{align*}
        &\quad\E\{\rho(Y_{1}^{(1)},Y_{1}^{(2)})G_{h_{1}}(X_{1}^{(1)}-X_{1}^{(2)})G_{h_{1}}(X_{2}^{(1)}-X_{1}^{(2)})G_{h_{2}}(X_{2}^{(2)}-X_{1}^{(1)})\}\\
        &\overset{(i)}{=}\int\rho(y_{1}^{(1)},y_{1}^{(2)})\prod_{s=1}^{p}\{g(u(s))g(v(s))g(w(s))\}\\
        &\quad\times f_{1}(y_{1}^{(1)},x+h_{1}u)f_{2}(y_{1}^{(2)},x)f_{1}(x+h_{1}v)f_{2}(x+h_{1}u+h_{2}w)dy_{1}^{(1)}dy_{1}^{(2)}dudvdwdx\\
        &=\int\E\{\rho(Y^{(1)},Y^{(2)})\mid X^{(1)}=x+h_{1}u,X^{(2)}=x\}\prod_{s=1}^{p}\{g(u(s))g(v(s))g(w(s))\}\\
        &\quad\times f_{1}(x+h_{1}u)f_{1}(x+h_{1}v)f_{2}(x+h_{1}u+h_{2}w)f_{2}(x)dudvdwdx\\
        &\overset{(ii)}{=}\int\E\{\rho(Y^{(1)},Y^{(2)})\mid X^{(1)}=x,X^{(2)}=x\}\prod_{s=1}^{p}\{g(u(s))g(v(s))g(w(s))\}\\
        &\quad\times \{f_{1}(x)f_{2}(x)\}^{2}dudvdwdx+O(h_{1}^{\nu}+h_{2}^{\nu})\\
        &=\int\E\{\rho(Y^{(1)},Y^{(2)})\mid X^{(1)}=x,X^{(2)}=x\}\{f_{1}(x)f_{2}(x)\}^{2}dx+O(h_{1}^{\nu}+h_{2}^{\nu}),
    \end{align*}
    where $(i)$ follows from change of variables, and $(ii)$ holds by Taylor's theorem and Assumptions \ref{assum1} and \ref{assum5}-\ref{assum6}. Similarly, one can verify that
    \begin{align*}
        &\quad\E\{\rho(Y_{1}^{(1)},Y_{1}^{(2)})G_{h_{2}}(X_{1}^{(2)}-X_{1}^{(1)})G_{h_{1}}(X_{2}^{(1)}-X_{1}^{(2)})G_{h_{2}}(X_{2}^{(2)}-X_{1}^{(1)})\}\\
        &=\int\E\{\rho(Y^{(1)},Y^{(2)})\mid X^{(1)}=x,X^{(2)}=x\}\{f_{1}(x)f_{2}(x)\}^{2}dx+O(h_{1}^{\nu}+h_{2}^{\nu}),\\
        &\quad\E\{\rho(Y_{1}^{(1)},Y_{2}^{(1)})G_{h_{1}}(X_{1}^{(1)}-X_{2}^{(1)})G_{h_{2}}(X_{1}^{(2)}-X_{1}^{(1)})G_{h_{2}}(X_{2}^{(2)}-X_{1}^{(1)})\}\\
        &=\int\E\{\rho(Y^{(1)},Y^{(1)\prime})\mid X^{(1)}=x,X^{(1)\prime}=x\}\{f_{1}(x)f_{2}(x)\}^{2}dx+O(h_{1}^{\nu}+h_{2}^{\nu}),\\
        &\quad\E\{\rho(Y_{1}^{(2)},Y_{2}^{(2)})G_{h_{2}}(X_{1}^{(2)}-X_{2}^{(2)})G_{h_{1}}(X_{1}^{(1)}-X_{1}^{(2)})G_{h_{1}}(X_{2}^{(1)}-X_{1}^{(2)})\}\\
        &=\int\E\{\rho(Y^{(2)},Y^{(2)\prime})\mid X^{(2)}=x,X^{(2)\prime}=x\}\{f_{1}(x)f_{2}(x)\}^{2}dx+O(h_{1}^{\nu}+h_{2}^{\nu}).
    \end{align*}
    Thus, by Assumption \ref{assum2}, 
    \begin{align*}
        \E(\widehat{\I}_{\rho})&=\int\Big[2\E\{\rho(Y^{(1)},Y^{(2)})\mid X^{(1)}=x,X^{(2)}=x\}-\E\{\rho(Y^{(1)},Y^{(1)\prime})\mid X^{(1)}=x,X^{(1)\prime}=x\}\\
        &\quad-\E\{\rho(Y^{(2)},Y^{(2)\prime})\mid X^{(2)}=x,X^{(2)\prime}=x\}\Big]\{f_{1}(x)f_{2}(x)\}^{2}dx+O(h_{1}^{\nu}+h_{2}^{\nu})\\
        &=\int\D_{\rho}(x)\{f_{1}(x)f_{2}(x)\}^{2}dx+o(1). 
    \end{align*}

    For the second part of (\ref{thm5_pf}), following Definition \ref{defa2} and Theorem \ref{thma1} ($m_{1}=m_{2}=2$), we have
    \[
    \var(\widehat{\I}_{\rho})=\binom{n_{1}}{2}^{-1}\binom{n_{2}}{2}^{-1}\sum_{c=0}^{2}\sum_{d=0}^{2}\binom{2}{c}\binom{2}{d}\binom{n_{1}-2}{2-c}\binom{n_{2}-2}{2-d}\sigma_{c,d}^{2}.
    \]
    We first calculate
    \[
    \sigma_{2,2}^{2}=\var\{\psi(Z_{1}^{(1)},Z_{2}^{(1)};Z_{1}^{(2)},Z_{2}^{(2)})\}=\E\{\psi^{2}(Z_{1}^{(1)},Z_{2}^{(1)};Z_{1}^{(2)},Z_{2}^{(2)})\}-\{\I_{\rho}+o(1)\}^{2}.
    \]
    As for $\E\{\psi^{2}(Z_{1}^{(1)},Z_{2}^{(1)};Z_{1}^{(2)},Z_{2}^{(2)})\}$, it can be expanded into several terms, and each of these terms can be shown to be of order $(h_{1}^{2}h_{2})^{-p}$ or $(h_{1}h_{2}^{2})^{-p}$. We only present the proof for the term below. Derivations for the other terms are similar. Note that
    \begin{align*}
        &\quad\E\Big[\Big\{\rho(Y_{1}^{(1)},Y_{1}^{(2)})G_{h_{1}}(X_{1}^{(1)}-X_{1}^{(2)})G_{h_{1}}(X_{2}^{(1)}-X_{1}^{(2)})G_{h_{2}}(X_{2}^{(2)}-X_{1}^{(1)})\Big\}^{2}\Big]\\
        &=\int\left\{\rho(y_{1}^{(1)},y_{1}^{(2)})G_{h_{1}}(x_{1}^{(1)}-x_{1}^{(2)})G_{h_{1}}(x_{2}^{(1)}-x_{1}^{(2)})G_{h_{2}}(x_{2}^{(2)}-x_{1}^{(1)})\right\}^{2}\\
        &\quad\times f_{1}(y_{1}^{(1)},x_{1}^{(1)})f_{2}(y_{1}^{(2)},x_{1}^{(2)})f_{1}(x_{2}^{(1)})f_{2}(x_{2}^{(2)})dy_{1}^{(1)}dy_{1}^{(2)}dx_{1}^{(1)}dx_{2}^{(1)}dx_{1}^{(2)}dx_{2}^{(2)}\\
        &\overset{(i)}{=}(h_{1}^{2}h_{2})^{-p}\int\bigg[\rho(y_{1}^{(1)},y_{1}^{(2)})\prod_{s=1}^{p}\{g(u(s))g(v(s))g(w(s))\}\bigg]^{2}\\
        &\quad\times f_{1}(y_{1}^{(1)},x+h_{1}u)f_{2}(y_{1}^{(2)},x)f_{1}(x+h_{1}v)f_{2}(x+h_{1}u+h_{2}w)dy_{1}^{(1)}dy_{1}^{(2)}dudvdwdx\\
        &\asymp(h_{1}^{2}h_{2})^{-p},
    \end{align*}
    where $(i)$ follows from the same change of variables technique as in the proof of the first part of (\ref{thm5_pf}) above. Hence, $\sigma_{2,2}^{2}=O((h_{1}^{2}h_{2})^{-p}+(h_{1}h_{2}^{2})^{-p})$. 

    Analogously, using change of variables, one can verify that $\sigma_{1,0}^{2}=O(1)$, $\sigma_{0,1}^{2}=O(1)$, $\sigma_{2,0}^{2}=O(h_{1}^{-p})$, $\sigma_{1,1}^{2}=O(h_{1}^{-p}+h_{2}^{-p})$, $\sigma_{0,2}^{2}=O(h_{2}^{-p})$, $\sigma_{2,1}^{2}=O(h_{1}^{-2p}+(h_{1}h_{2})^{-p})$, and $\sigma_{1,2}^{2}=O((h_{1}h_{2})^{-p}+h_{2}^{-2p})$. Therefore, by Assumption \ref{assum2}, 
    \[
    \var(\widehat{\I}_{\rho})=\sum_{c=0}^{2}\sum_{d=0}^{2}O(n_{1}^{-c}n_{2}^{-d})\sigma_{c,d}^{2}=O(n_{1}^{-1}+n_{2}^{-1})=o(1).
    \]
    This completes the proof. 
\end{proof}

\begin{proof}[Proof of Theorem \ref{thm6}]
    Following Definition \ref{defa3} and Theorem \ref{thma2}, we have the Hoeffding decomposition for the U-statistic $\widehat{\I}_{\rho}$:
    \[
    \widehat{\I}_{\rho}=\sum_{c=0}^{2}\sum_{d=0}^{2}\binom{2}{c}\binom{2}{d}H_{n_{1},n_{2}}^{(c,d)}=\E(\widehat{\I}_{\rho})+2H_{n_{1},n_{2}}^{(1,0)}+2H_{n_{1},n_{2}}^{(0,1)}+R_{n_{1},n_{2}}, 
    \]
    where
    \begin{align*}
        H_{n_{1},n_{2}}^{(1,0)}&=\frac{1}{n_{1}}\sum_{i=1}^{n_{1}}\phi^{(1,0)}(Z_{i}^{(1)}),\quad H_{n_{1},n_{2}}^{(0,1)}=\frac{1}{n_{2}}\sum_{j=1}^{n_{2}}\phi^{(0,1)}(Z_{j}^{(2)}),
    \end{align*}
    and $R_{n_{1},n_{2}}$ is the remainder term. Furthermore, we give the explicit expression of $\phi^{(1,0)}(z_{1}^{(1)})$ (the expression of $\phi^{(0,1)}(z_{1}^{(2)})$ is similar and omitted):
    \begin{align}
        \begin{split}\label{eq:hoeff2}
            &\quad\phi^{(1,0)}(z_{1}^{(1)})\\
            &=\frac{1}{2}\E\Big[\rho(y_{1}^{(1)},Y_{1}^{(2)})\left\{G_{h_{1}}(x_{1}^{(1)}-X_{1}^{(2)})+G_{h_{2}}(X_{1}^{(2)}-x_{1}^{(1)})\right\}G_{h_{1}}(X_{2}^{(1)}-X_{1}^{(2)})G_{h_{2}}(X_{2}^{(2)}-x_{1}^{(1)})\Big]\\
            &\quad+\frac{1}{2}\E\Big[\rho(Y_{2}^{(1)},Y_{1}^{(2)})\left\{G_{h_{1}}(X_{2}^{(1)}-X_{1}^{(2)})+G_{h_{2}}(X_{1}^{(2)}-X_{2}^{(1)})\right\}G_{h_{1}}(x_{1}^{(1)}-X_{1}^{(2)})G_{h_{2}}(X_{2}^{(2)}-X_{2}^{(1)})\Big]\\
            &\quad-\frac{1}{2}\E\Big\{\rho(y_{1}^{(1)},Y_{2}^{(1)})G_{h_{1}}(x_{1}^{(1)}-X_{2}^{(1)})G_{h_{2}}(X_{1}^{(2)}-x_{1}^{(1)})G_{h_{2}}(X_{2}^{(2)}-x_{1}^{(1)})\Big\}\\
            &\quad-\frac{1}{2}\E\Big\{\rho(y_{1}^{(1)},Y_{2}^{(1)})G_{h_{1}}(x_{1}^{(1)}-X_{2}^{(1)})G_{h_{2}}(X_{1}^{(2)}-X_{2}^{(1)})G_{h_{2}}(X_{2}^{(2)}-X_{2}^{(1)})\Big\}\\
            &\quad-\E\Big\{\rho(Y_{1}^{(2)},Y_{2}^{(2)})G_{h_{2}}(X_{1}^{(2)}-X_{2}^{(2)})G_{h_{1}}(x_{1}^{(1)}-X_{1}^{(2)})G_{h_{1}}(X_{2}^{(1)}-X_{1}^{(2)})\Big\}\\
            &\quad-\E(\widehat{\I}_{\rho}).
        \end{split}
    \end{align}

    By similar calculations as in the proof of Theorem \ref{thm5}, we have 
    \begin{align}
        \delta_{1,0}^{2}&=\var\{\phi^{(1,0)}(Z_{1}^{(1)})\}=\sigma_{1,0}^{2}\asymp 1, \label{eq:delta10}\\
        \delta_{0,1}^{2}&=\var\{\phi^{(0,1)}(Z_{1}^{(2)})\}=\sigma_{0,1}^{2}\asymp 1, \label{eq:delta01}
    \end{align}
    under $H_{a}$ in (\ref{eq:hypo}). Applying Lyapunov CLT, 
    \[
    \frac{H_{n_{1},n_{2}}^{(1,0)}+H_{n_{1},n_{2}}^{(0,1)}}{\sqrt{\frac{\delta_{1,0}^{2}}{n_{1}}+\frac{\delta_{0,1}^{2}}{n_{2}}}}\convd\N(0,1).
    \]

    Also, one can show that
    \begin{align*}
        \var(H_{n_{1},n_{2}}^{(2,0)})&\asymp n_{1}^{-2}\delta_{2,0}^{2}\asymp n_{1}^{-2}\sigma_{2,0}^{2}\asymp (n_{1}^{2}h_{1}^{p})^{-1},\\
        \var(H_{n_{1},n_{2}}^{(1,1)})&\asymp (n_{1}n_{2})^{-1}\delta_{1,1}^{2}\asymp (n_{1}n_{2})^{-1}\sigma_{1,1}^{2}\asymp (n_{1}h_{1}^{p})^{-1}n_{2}^{-1}+n_{1}^{-1}(n_{2}h_{2}^{p})^{-1},\\
        \var(H_{n_{1},n_{2}}^{(0,2)})&\asymp n_{2}^{-2}\delta_{0,2}^{2}\asymp n_{2}^{-2}\sigma_{0,2}^{2}\asymp (n_{2}^{2}h_{2}^{p})^{-1},\\
        \var(H_{n_{1},n_{2}}^{(2,1)})&\asymp n_{1}^{-2}n_{2}^{-1}\delta_{2,1}^{2}\asymp n_{1}^{-2}n_{2}^{-1}\sigma_{2,1}^{2}\asymp (n_{1}h_{1}^{p})^{-2}n_{2}^{-1}+(n_{1}^{2}h_{1}^{p})^{-1}(n_{2}h_{2}^{p})^{-1},\\
        \var(H_{n_{1},n_{2}}^{(1,2)})&\asymp n_{1}^{-1}n_{2}^{-2}\delta_{1,2}^{2}\asymp n_{1}^{-1}n_{2}^{-2}\sigma_{1,2}^{2}\asymp (n_{1}h_{1}^{p})^{-1}(n_{2}^{2}h_{2}^{p})^{-1}+n_{1}^{-1}(n_{2}h_{2}^{p})^{-2},\\
        \var(H_{n_{1},n_{2}}^{(2,2)})&\asymp (n_{1}n_{2})^{-2}\delta_{2,2}^{2}\asymp (n_{1}n_{2})^{-2}\sigma_{2,2}^{2}\asymp (n_{1}h_{1}^{p})^{-2}(n_{2}^{2}h_{2}^{p})^{-1}+(n_{1}^{2}h_{1}^{p})^{-1}(n_{2}h_{2}^{p})^{-2}.
    \end{align*}
    Thus, 
    \[
    \frac{\var(R_{n_{1},n_{2}})}{\frac{\delta_{1,0}^{2}}{n_{1}}+\frac{\delta_{0,1}^{2}}{n_{2}}}=O((n_{1}h_{1}^{p})^{-1}+(n_{2}h_{2}^{p})^{-1}).
    \]
    Then we have
    \begin{align*}
        \frac{\widehat{\I}_{\rho}-\I_{\rho}}{\sqrt{\frac{\delta_{1,0}^{2}}{n_{1}}+\frac{\delta_{0,1}^{2}}{n_{2}}}}&=\frac{2H_{n_{1},n_{2}}^{(1,0)}+2H_{n_{1},n_{2}}^{(0,1)}}{\sqrt{\frac{\delta_{1,0}^{2}}{n_{1}}+\frac{\delta_{0,1}^{2}}{n_{2}}}}+\frac{\E(\widehat{\I}_{\rho})-\I_{\rho}}{\sqrt{\frac{\delta_{1,0}^{2}}{n_{1}}+\frac{\delta_{0,1}^{2}}{n_{2}}}}+\frac{R_{n_{1},n_{2}}}{\sqrt{\frac{\delta_{1,0}^{2}}{n_{1}}+\frac{\delta_{0,1}^{2}}{n_{2}}}}\\
        &=\frac{2H_{n_{1},n_{2}}^{(1,0)}+2H_{n_{1},n_{2}}^{(0,1)}}{\sqrt{\frac{\delta_{1,0}^{2}}{n_{1}}+\frac{\delta_{0,1}^{2}}{n_{2}}}}+O(n_{1}^{1/2}h_{1}^{\nu}+n_{2}^{1/2}h_{2}^{\nu})+O_{p}((n_{1}h_{1}^{p})^{-1/2}+(n_{2}h_{2}^{p})^{-1/2})\\
        &=\frac{2H_{n_{1},n_{2}}^{(1,0)}+2H_{n_{1},n_{2}}^{(0,1)}}{\sqrt{\frac{\delta_{1,0}^{2}}{n_{1}}+\frac{\delta_{0,1}^{2}}{n_{2}}}}+o_{p}(1)\convd\N(0,4),
    \end{align*}
    where we have used Assumption \ref{assum2} and $n_{l}^{1/2}h_{l}^{\nu}\rightarrow 0$ as $n_{l}\rightarrow\infty$ for $l=1,2$. 
\end{proof}

\begin{proof}[Proof of Theorem \ref{thm7}]
    Following the proof of Theorem \ref{thm6}, we have the Hoeffding decomposition
    \[
    \widehat{\I}_{\rho}=\E(\widehat{\I}_{\rho})+2H_{n_{1},n_{2}}^{(1,0)}+2H_{n_{1},n_{2}}^{(0,1)}+H_{n_{1},n_{2}}^{(2,0)}+4H_{n_{1},n_{2}}^{(1,1)}+H_{n_{1},n_{2}}^{(0,2)}+R_{n_{1},n_{2}},
    \]
    where $R_{n_{1},n_{2}}$ is the remainder term, with slight abuse of notation. 

    Under $H_{0}$ in (\ref{eq:hypo}), one can verify that $\phi^{(1,0)}(z_{1}^{(1)})\neq 0$ and $\phi^{(0,1)}(z_{1}^{(2)})\neq 0$, and thus the U-statistic $\widehat{\I}_{\rho}$ is nondegenerate. Nonetheless, as in the local case, $H_{n_{1},n_{2}}^{(1,0)}$ and $H_{n_{1},n_{2}}^{(0,1)}$ can be asymptotically negligible under $H_{0}$ in (\ref{eq:hypo}). Specifically, by similar calculations as in the proof of Theorem \ref{thm5}, 
    \begin{align*}
        &\quad\E\Big[\rho(y_{1}^{(1)},Y_{1}^{(2)})\left\{G_{h_{1}}(x_{1}^{(1)}-X_{1}^{(2)})+G_{h_{2}}(X_{1}^{(2)}-x_{1}^{(1)})\right\}G_{h_{1}}(X_{2}^{(1)}-X_{1}^{(2)})G_{h_{2}}(X_{2}^{(2)}-x_{1}^{(1)})\Big]\\
        &=2f_{1}(x_{1}^{(1)})f_{2}^{2}(x_{1}^{(1)})\int\rho(y_{1}^{(1)},y_{1}^{(2)})f_{2}(y_{1}^{(2)}\mid x_{1}^{(1)})dy_{1}^{(2)}\times\{1+O(h_{1}^{\nu}+h_{2}^{\nu})\},\\
        &\quad\E\Big[\rho(Y_{2}^{(1)},Y_{1}^{(2)})\left\{G_{h_{1}}(X_{2}^{(1)}-X_{1}^{(2)})+G_{h_{2}}(X_{1}^{(2)}-X_{2}^{(1)})\right\}G_{h_{1}}(x_{1}^{(1)}-X_{1}^{(2)})G_{h_{2}}(X_{2}^{(2)}-X_{2}^{(1)})\Big]\\
        &=2f_{1}(x_{1}^{(1)})f_{2}^{2}(x_{1}^{(1)})\int\rho(y_{2}^{(1)},y_{1}^{(2)})f_{1}(y_{2}^{(1)}\mid x_{1}^{(1)})f_{2}(y_{1}^{(2)}\mid x_{1}^{(1)})dy_{2}^{(1)}dy_{1}^{(2)}\times\{1+O(h_{1}^{\nu}+h_{2}^{\nu})\},\\
        &\quad\E\Big\{\rho(y_{1}^{(1)},Y_{2}^{(1)})G_{h_{1}}(x_{1}^{(1)}-X_{2}^{(1)})G_{h_{2}}(X_{1}^{(2)}-x_{1}^{(1)})G_{h_{2}}(X_{2}^{(2)}-x_{1}^{(1)})\Big\}\\
        &=f_{1}(x_{1}^{(1)})f_{2}^{2}(x_{1}^{(1)})\int\rho(y_{1}^{(1)},y_{2}^{(1)})f_{1}(y_{2}^{(1)}\mid x_{1}^{(1)})dy_{2}^{(1)}\times\{1+O(h_{1}^{\nu}+h_{2}^{\nu})\},\\
        &\quad\E\Big\{\rho(y_{1}^{(1)},Y_{2}^{(1)})G_{h_{1}}(x_{1}^{(1)}-X_{2}^{(1)})G_{h_{2}}(X_{1}^{(2)}-X_{2}^{(1)})G_{h_{2}}(X_{2}^{(2)}-X_{2}^{(1)})\Big\}\\
        &=f_{1}(x_{1}^{(1)})f_{2}^{2}(x_{1}^{(1)})\int\rho(y_{1}^{(1)},y_{2}^{(1)})f_{1}(y_{2}^{(1)}\mid x_{1}^{(1)})dy_{2}^{(1)}\times\{1+O(h_{1}^{\nu}+h_{2}^{\nu})\},\\
        &\quad\E\Big\{\rho(Y_{1}^{(2)},Y_{2}^{(2)})G_{h_{2}}(X_{1}^{(2)}-X_{2}^{(2)})G_{h_{1}}(x_{1}^{(1)}-X_{1}^{(2)})G_{h_{1}}(X_{2}^{(1)}-X_{1}^{(2)})\Big\}\\
        &=f_{1}(x_{1}^{(1)})f_{2}^{2}(x_{1}^{(1)})\int\rho(y_{1}^{(2)},y_{2}^{(2)})f_{2}(y_{1}^{(2)}\mid x_{1}^{(1)})f_{2}(y_{2}^{(2)}\mid x_{1}^{(1)})dy_{1}^{(2)}dy_{2}^{(2)}\times\{1+O(h_{1}^{\nu}+h_{2}^{\nu})\}.              
    \end{align*}
    Under $H_{0}$ in (\ref{eq:hypo}), we have $f(y\mid x)\coloneqq f_{1}(y\mid x)=f_{2}(y\mid x)$ for all $y,x$. Hence, for $\phi^{(1,0)}(z_{1}^{(1)})$ in (\ref{eq:hoeff2}), 
    \begin{align*}
        \phi^{(1,0)}(z_{1}^{(1)})&=f_{1}(x_{1}^{(1)})f_{2}^{2}(x_{1}^{(1)})\bigg\{\int\rho(y_{1}^{(1)},y)f(y\mid x_{1}^{(1)})dy+\int\rho(y,y')f(y\mid x_{1}^{(1)})f(y'\mid x_{1}^{(1)})dydy'\bigg\}\\
        &\quad\times O(h_{1}^{\nu}+h_{2}^{\nu})+O(h_{1}^{\nu}+h_{2}^{\nu}),
    \end{align*}
    under $H_{0}$ in (\ref{eq:hypo}). It implies that
    \[
    \var\{\phi^{(1,0)}(Z_{1}^{(1)})\}=O(h_{1}^{2\nu}+h_{2}^{2\nu}),
    \]
    and thus, 
    \[
    \var(H_{n_{1},n_{2}}^{(1,0)})=O(n_{1}^{-1}(h_{1}^{2\nu}+h_{2}^{2\nu})).
    \]
    Analogously, one can show that $\var(H_{n_{1},n_{2}}^{(0,1)})=O(n_{2}^{-1}(h_{1}^{2\nu}+h_{2}^{2\nu}))$ under $H_{0}$ in (\ref{eq:hypo}). As in the local case, when undersmoothing is used, the asymptotic distribution of $\widehat{\I}_{\rho}$ is determined by $H_{n_{1},n_{2}}^{(2,0)}+4H_{n_{1},n_{2}}^{(1,1)}+H_{n_{1},n_{2}}^{(0,2)}$ under $H_{0}$ in (\ref{eq:hypo}), since $\var(H_{n_{1},n_{2}}^{(2,0)})\asymp (n_{1}^{2}h_{1}^{p})^{-1}$, $\var(H_{n_{1},n_{2}}^{(1,1)})\asymp (n_{1}h_{1}^{p})^{-1}n_{2}^{-1}+n_{1}^{-1}(n_{2}h_{2}^{p})^{-1}$ and $\var(H_{n_{1},n_{2}}^{(0,2)})\asymp (n_{2}^{2}h_{2}^{p})^{-1}$. 
    
    We have
    \begin{align*}
        H_{n_{1},n_{2}}^{(2,0)}+4H_{n_{1},n_{2}}^{(1,1)}+H_{n_{1},n_{2}}^{(0,2)}&=\binom{n_{1}}{2}^{-1}\sum_{1\le i_{1}<i_{2}\le n_{1}}\phi^{(2,0)}(Z_{i_{1}}^{(1)},Z_{i_{2}}^{(1)})+\frac{4}{n_{1}n_{2}}\sum_{i=1}^{n_{1}}\sum_{j=1}^{n_{2}}\phi^{(1,1)}(Z_{i}^{(1)};Z_{j}^{(2)})\\
        &\quad+\binom{n_{2}}{2}^{-1}\sum_{1\le j_{1}<j_{2}\le n_{2}}\phi^{(0,2)}(Z_{j_{1}}^{(2)},Z_{j_{2}}^{(2)}),
    \end{align*}
    and aim to show that
    \[
    \sigma_{n_{1},n_{2}}^{-1}(H_{n_{1},n_{2}}^{(2,0)}+4H_{n_{1},n_{2}}^{(1,1)}+H_{n_{1},n_{2}}^{(0,2)})\convd\N(0,1),\quad\text{as }n_{1},n_{2}\rightarrow\infty,
    \]
    where
    \[
    \sigma_{n_{1},n_{2}}^{2}\coloneqq\var(H_{n_{1},n_{2}}^{(2,0)}+4H_{n_{1},n_{2}}^{(1,1)}+H_{n_{1},n_{2}}^{(0,2)})=\frac{2}{n_{1}(n_{1}-1)}\delta_{2,0}^{2}+\frac{16}{n_{1}n_{2}}\delta_{1,1}^{2}+\frac{2}{n_{2}(n_{2}-1)}\delta_{0,2}^{2},
    \]
    and
    \begin{align}
        \delta_{2,0}^{2}&=\var\{\phi^{(2,0)}(Z_{1}^{(1)},Z_{2}^{(1)})\}, \label{eq:delta20}\\
        \delta_{1,1}^{2}&=\var\{\phi^{(1,1)}(Z_{1}^{(1)};Z_{1}^{(2)})\}, \label{eq:delta11}\\
        \delta_{0,2}^{2}&=\var\{\phi^{(0,2)}(Z_{1}^{(2)},Z_{2}^{(2)})\}. \label{eq:delta02}
    \end{align}
    
    Let $W_{i}=Z_{i}^{(1)}$ for $i=1,\ldots,n_{1}$ and $W_{j+n_{1}}=Z_{j}^{(2)}$ for $j=1,\ldots,n_{2}$, and for $i\neq j$,
    \begin{align*}
        \varphi_{ij}=\left\{
        \begin{aligned}
            n_{1}^{-1}(n_{1}-1)^{-1}\phi^{(2,0)}(W_{i},W_{j}),\quad&\text{if }i,j\in\{1,\ldots,n_{1}\},\\
            2n_{1}^{-1}n_{2}^{-1}\phi^{(1,1)}(W_{i},W_{j}),\quad&\text{if }i\in\{1,\ldots,n_{1}\}\text{ and }j\in\{n_{1}+1,\ldots,n_{1}+n_{2}\},\\
            n_{2}^{-1}(n_{2}-1)^{-1}\phi^{(0,2)}(W_{i},W_{j}),\quad&\text{if }i,j\in\{n_{1}+1,\ldots,n_{1}+n_{2}\}.
        \end{aligned}
        \right.
    \end{align*}
    Write $n=n_{1}+n_{2}$. Define $V_{n,j}=\sum_{i=1}^{j-1}\varphi_{ij}$ for $j=2,\ldots,n$, and $S_{n,m}=\sum_{j=2}^{m}V_{n,j}$ for $m=2,\ldots,n$. Let $\F_{n,m}$ be the $\sigma$-algebra generated by $\{W_{1},\ldots,W_{m}\}$ for $m=1,\ldots,n$. It is straightforward that $\F_{n,m-1}\subset\F_{n,m}$ for any $m=2,\ldots,n$, and $\{S_{n,m}\}_{m=1}^{n}$ is of zero mean and square integrable. Also, for $m_{1}>m_{2}$, we have $\E(S_{n,m_{1}}\mid \F_{n,m_{2}})=S_{n,m_{2}}$. Hence, for each $n$, $\{S_{n,m},\F_{n,m}\}_{m=1}^{n}$ is a square integrable martingale of zero mean. As
    \[
    H_{n_{1},n_{2}}^{(2,0)}+4H_{n_{1},n_{2}}^{(1,1)}+H_{n_{1},n_{2}}^{(0,2)}=2S_{n,n},
    \]
    its asymptotic normality can be established by Corollary 3.1 in \citet{hall2014martingale}, i.e.,
    \begin{align*}
        \sigma_{n_{1},n_{2}}^{-2}\sum_{j=2}^{n}\E(V_{n,j}^{2}\mid\F_{n,j-1})&\convp\frac{1}{4},\\
        \sigma_{n_{1},n_{2}}^{-2}\sum_{j=2}^{n}\E[V_{n,j}^{2}\mathbbm{1}\{|V_{n,j}|>\epsilon\sigma_{n_{1},n_{2}}\}\mid\F_{n,j-1}]&\convp0,\quad\forall\epsilon>0,
    \end{align*}
    which can be done with routine verification of the following conditions:
    \begin{align*}
        \var\bigg\{\sum_{j=2}^{n}\E(V_{n,j}^{2}\mid\F_{n,j-1})\bigg\}&=o(\sigma_{n_{1},n_{2}}^{4}),\\
        \sum_{j=2}^{n}\E(V_{n,j}^{4})&=o(\sigma_{n_{1},n_{2}}^{4}).
    \end{align*}
    In our case, we have
    \begin{align*}
        \sigma_{n_{1},n_{2}}^{-4}\var\bigg\{\sum_{j=2}^{n}\E(V_{n,j}^{2}\mid\F_{n,j-1})\bigg\}&=O(h_{1}^{p}+h_{2}^{p})=o(1),\\
        \sigma_{n_{1},n_{2}}^{-4}\sum_{j=2}^{n}\E(V_{n,j}^{4})&=O((n_{1}h_{1}^{p}+n_{2}h_{2}^{p})^{-1})=o(1),
    \end{align*}
    where we have used Assumption \ref{assum2}. The tedious calculations are omitted here. 

    Therefore, 
    \begin{align*}
        \sigma_{n_{1},n_{2}}^{-1}\widehat{\I}_{\rho}&=\sigma_{n_{1},n_{2}}^{-1}(H_{n_{1},n_{2}}^{(2,0)}+4H_{n_{1},n_{2}}^{(1,1)}+H_{n_{1},n_{2}}^{(0,2)})+\sigma_{n_{1},n_{2}}^{-1}\E(\widehat{\I}_{\rho})\\
        &\quad+2\sigma_{n_{1},n_{2}}^{-1}(H_{n_{1},n_{2}}^{(1,0)}+H_{n_{1},n_{2}}^{(0,1)})+\sigma_{n_{1},n_{2}}^{-1}R_{n_{1},n_{2}}\\
        &=\sigma_{n_{1},n_{2}}^{-1}(H_{n_{1},n_{2}}^{(2,0)}+4H_{n_{1},n_{2}}^{(1,1)}+H_{n_{1},n_{2}}^{(0,2)})+O(n_{1}h_{1}^{p/2+\nu}+n_{2}h_{2}^{p/2+\nu})\\
        &\quad+O_{p}(n_{1}^{1/2}h_{1}^{p/2+\nu}+n_{2}^{1/2}h_{2}^{p/2+\nu})+O_{p}((n_{1}h_{1}^{p})^{-1/2}+(n_{2}h_{2}^{p})^{-1/2})\\
        &=\sigma_{n_{1},n_{2}}^{-1}(H_{n_{1},n_{2}}^{(2,0)}+4H_{n_{1},n_{2}}^{(1,1)}+H_{n_{1},n_{2}}^{(0,2)})+o_{p}(1)\convd\N(0,1),
    \end{align*}
    where we have used Assumption \ref{assum2} and $n_{l}h_{l}^{p/2+\nu}\rightarrow 0$ as $n_{l}\rightarrow\infty$ for $l=1,2$. 
\end{proof}

\end{appendices}

\newpage
\bibliography{main}

\end{document}